\documentclass[journal]{IEEEtran}
\usepackage{cite}
\ifCLASSINFOpdf
  \usepackage[pdftex]{graphicx}
\else
\fi

\usepackage{amsmath}
\usepackage{amsthm}
\usepackage{amssymb}

\ifCLASSOPTIONcompsoc
  \usepackage[caption=false,font=normalsize,labelfont=sf,textfont=sf]{subfig}
\else
  \usepackage[caption=false,font=footnotesize]{subfig}
\fi

\usepackage{url}

\newtheorem{defn}{Definition}
\newtheorem{theorem}{Theorem}
\newtheorem{lemma}{Lemma}
\newtheorem{proposition}{Proposition}

\newtheorem{remark}{Remark}
\newcommand{\bu}{\mathbf{u}}
\newcommand{\bU}{\mathbf{U}}

\newcommand{\NN}{\mathcal{N}}
\newcommand{\fh}[1]{FH-SP$_{#1}$}
\newcommand{\ih}{\ref{eq:i-miqp}}
\usepackage{mathtools}
\usepackage[binary-units=true]{siunitx}
\newcommand{\ceil}[1]{\left\lceil{#1}\right\rceil}
\DeclareMathOperator*{\argmin}{arg\,min}

\usepackage{tikz}
\usepackage{tikzscale}
\usetikzlibrary{shapes,shapes.geometric,arrows,positioning,automata,fit,intersections}
\usetikzlibrary{external}
\newcommand{\includegtikz}[2][]{ 
	\tikzsetnextfilename{#2}
	\includegraphics[#1]{figures/tikz/#2} }

\usepackage{graphicx} 
\usepackage{multirow}
\usepackage{pgfplots}
\usepgfplotslibrary{fillbetween}
\renewcommand{\parallel}{{\mkern3mu\vphantom{\perp}\vrule depth 0pt\mkern2mu\vrule depth 0pt\mkern3mu}}

\newif\ifapp
\apptrue

\usepackage{todonotes}

\usepackage{cleveref}
\crefname{appsec}{Appendix}{Appendices}
\crefname{theorem}{theorem}{theorems}
\crefname{lemma}{lemma}{lemmas}
\crefname{remark}{remark}{remarks}
\crefname{proposition}{proposition}{propositions}
\crefname{corollary}{corollary}{corollaries}
\crefname{defn}{definition}{definitions}
\begin{document}
\title{An Algorithm for Supervised Driving of Cooperative Semi-Autonomous Vehicles (Extended)}

\author{Florent~Altch\'e,
        Xiangjun~Qian,
        and~Arnaud~de~La~Fortelle,
\thanks{F. Altch\'e, X. Qian and A. de La Fortelle are with MINES ParisTech, PSL Research University, Centre for robotics, 60 Bd St Michel 75006 Paris, France (e-mail: \texttt{florent.altche@mines-paristech.fr}.)}
\thanks{F. Altch\'e is also with \'Ecole des Ponts ParisTech, Cit\'e Descartes, 6-8 Av Blaise Pascal, 77455 Champs-sur-Marne, France.}}

\maketitle

\begin{abstract}
Before reaching full autonomy, vehicles will gradually be equipped with more and more advanced driver assistance systems (ADAS), effectively rendering them semi-autonomous. However, current ADAS technologies seem unable to handle complex traffic situations, notably when dealing with vehicles arriving from the sides, either at intersections or when merging on highways. The high rate of accidents in these settings prove that they constitute difficult driving situations. Moreover, intersections and merging lanes are often the source of important traffic congestion and, sometimes, deadlocks. In this article, we propose a cooperative framework to safely coordinate semi-autonomous vehicles in such settings, removing the risk of collision or deadlocks while remaining compatible with human driving. More specifically, we present a \textit{supervised coordination} scheme that overrides control inputs from human drivers when they would result in an unsafe or blocked situation. To avoid unnecessary intervention and remain compatible with human driving, overriding only occurs when collisions or deadlocks are imminent. In this case, safe overriding controls are chosen while ensuring they deviate minimally from those originally requested by the drivers. Simulation results based on a realistic physics simulator show that our approach is scalable to real-world scenarios, and computations can be performed in real-time on a standard computer for up to a dozen simultaneous vehicles.
\end{abstract}

\begin{IEEEkeywords}
Semi-autonomous driving, safety, supervisor, supervised driving.
\end{IEEEkeywords}

%
\IEEEpeerreviewmaketitle

\section{Introduction}

\IEEEPARstart{A}{dvanced} driver assistance systems (ADAS) are becoming increasingly complex as they spread across the automotive market. Although adaptive cruise control (ACC)~\cite{Vahidi2003} and automated emergency braking (AEB)~\cite{Coelingh2010} are the best-known examples of such systems, applications of ADAS have been broadened and now include pedestrian~\cite{Geronimo2010}, traffic light~\cite{Diaz-Cabrera2015} or obstacle detection~\cite{Kim2015} as well as lane keeping assistance~\cite{Kim2008}. The development of this new equipment allows drivers to delegate part of the driving task to their vehicles. As these systems keep getting more efficient and able to handle more complex situations, vehicles will gradually progress towards \textit{semi-autonomous} driving, where drivers remain in charge of their own safety, while their errors can be seamlessly corrected to prevent potential accidents.

One of the challenges of semi-autonomous driving lies in efficiently handling vehicles on conflicting paths, for instance at an intersection or a highway entry lane. Traffic rules such as priority to the right can help determine whether to pass before or after another vehicle; however, many situations require driving experience to be handled efficiently. Learning-based approaches may eventually prove able to transfer driving experience to a computer, but such knowledge is very hard to implement in a safety system. In this article, we consider another possible solution, consisting in using vehicle-to-vehicle or vehicle-to-infrastructure communication for cooperative semi-autonomous driving. In this setting, vehicles negotiate with one another, or receive instructions from a centralized computer, allowing them to drive safely and efficiently.

In this article, we consider a method to ensure the safety of multiple semi-autonomous vehicles on conflicting paths, for instance crossing an intersection or entering a highway, while remaining compatible with the presence of human drivers. To this end, and inspired by earlier work in~\cite{Colombo2014,Campos2015}, we propose a so-called \textit{Supervisor} which monitors control inputs from each vehicle's driver, and is able to override these controls when they would result in an unsafe situation. More specifically, the role of the supervisor is twofold: first, knowing the current states of the vehicles, the supervisor should determine if the controls requested by the drivers would lead the vehicles into unsafe \textit{inevitable collision states}~\cite{Fraichard2003}. In this case, the second task of the supervisor is to compute safe controls -- maintaining the vehicles in safe states -- which are as close as possible to those actually requested by the drivers. We say that such a control is \textit{minimally deviating}.

This paper provides two main contributions: from a practical standpoint, we design and implement a mathematical framework allowing to simultaneously perform the safety verification of target control inputs, and the computation of minimally deviating safe controls if target inputs are unsafe. From a theoretical standpoint, we formally prove that verifying safety over a finite time horizon is enough to ensure infinite horizon safety, and we provide a sufficient condition on the verification horizon for this property to hold. Unlike previous work focusing on specific situations such as intersections~\cite{Colombo2014,Campos2015}, our framework can be applied to a wide variety of driving scenarios including intersections, merging lanes and roundabouts.

The rest of the article is structured as follows: in \Cref{sec:related-work}, we provide a review of the related literature. In \Cref{sec:problem-statement}, we present our modeling of semi-autonomous vehicles and introduce the \textit{Supervision problem} of verifying the safety of drivers control inputs and finding a minimally deviating safe control if necessary. In \Cref{sec:miqp-supervised}, we present an infinite horizon formulation based on constraints programming to solve this problem. In \Cref{sec:finite-horizon}, we derive a finite horizon formulation which we prove is equivalent to the infinite horizon one. In \Cref{sec:simulation}, we use computer simulations to showcase the performance of the proposed supervisor in various driving situations. In \Cref{sec:discussion}, we present possible methods for real-world implementations of our approach. Finally, \Cref{sec:conclusion} concludes the study.

\section{Related Work}\label{sec:related-work}
In the last decade, a lot of research has been focused on coordinating fully autonomous vehicles in challenging settings such as crossroads, roundabouts or merging lanes, with the ambition of improving both safety and traffic efficiency. Naumann et al.~\cite{naumann1998managing}, followed by Dresner and Stone~\cite{Dresner2004} have seemingly pioneered the work of adapting traffic intersections management methods to fully autonomous vehicles, designing so-called \textit{autonomous intersection management} algorithms. They propose that each approaching autonomous vehicle reserves a time interval to cross the intersection; collisions are prevented by ensuring that conflicting vehicles are assigned non-overlapping crossing times. Subsequent studies on this particular problem have led to other approaches. In~\cite{Makarem2011}, vehicles choose their control inputs based on navigation functions which include a collision avoidance term, allowing vehicles to react to maneuvers from other traffic participants. In~\cite{Qian2015}, collision avoidance is ensured by assigning relative crossing orders to incoming vehicles; each vehicle then uses model predictive control to plan collision-free trajectories respecting these priorities. Other authors have considered different driving situations for autonomous vehicles, such as cooperative merging on a highway~\cite{Mu2015,Ntousakis2016,Mosebach2016}, or entering a roundabout~\cite{Desaraju2009}.

By contrast, relatively little work has considered semi-autonomous driving assistance, possibly because the presence of human drivers brings a lot of additional complexity. The goal of a semi-autonomous driving assistant is to help the driver avoid collisions, either by notifying of a potential danger~\cite{Vasudevan2012} or by taking over vehicle control in dangerous situations~\cite{Gray2013,Gray2013a,Liu2014}. To be accepted by human drivers, such systems should be as unobtrusive as possible, and in particular should only intervene when necessary. Most of the currently existing literature on semi-autonomous driving mostly focuses on highway driving~\cite{Gray2013,Gray2013a,Liu2014}, which presents relatively low difficulty as vehicles trajectories remain mostly parallel. The aim of this article is to bring semi-autonomy one step further, to allow cooperative driving between semi-autonomous vehicles in more complex conflict situations.

Some of these more complex problems have already been studied in the literature. In~\cite{Au2014}, the authors consider \textit{semi-autonomous} driving at an intersection and propose that human drivers let an automated system control their vehicle while crossing said intersection. However, this scheme is rather intrusive as drivers completely relinquish control for a time, and handing back controls to a potentially distracted driver poses problems by itself. Colombo et al.~\cite{Colombo2014,Campos2015} introduced the idea of a supervisory instance (called \textit{supervisor}) tasked with preventing the system of vehicles from entering undesirable states by overriding the controls of one or several vehicles. In this more human-friendly approach, overriding  only occurs when necessary, \textit{i.e.} if an absence of intervention would result in a crash. The question of determining whether overriding is needed or not, called \textit{verification problem}, is NP-hard~\cite{Reveliotis2010}; under several simplifying assumptions, it is shown in~\cite{Colombo2014} to be equivalent to a \textit{scheduling problem}. In this reformulation, vehicles are each assigned a time slot during which they are allowed inside the intersection, and assigned slots are mutually disjoint. If, due to vehicle dynamics, no feasible schedule exists, the initial state is deemed unsafe. This allows the authors to design a so-called \textit{least restrictive} supervisor, which verifies the safety of the desired inputs and overrides them if necessary. However, the proposed supervisor is only suitable to simple intersection geometries with a single conflict point. Moreover, no additional property is required from the safe controls used for overriding, which can widely deviate from the desired ones.

Several variations have been proposed based on the equivalence demonstrated in~\cite{Colombo2014}. Reference~\cite{Bruni2013} designs a supervisor which is robust to bounded uncertainties by adding safety margins. Reference~\cite{Ahn2015} leverages job-shop scheduling to develop a supervisor that considers several possible conflict points inside the intersection; however, vehicle dynamics are only modeled as first-order integrators, which is not realistic in a real-world setting. Campos et al.~\cite{Campos2015} proposed a Pareto-optimal supervisor leading to a minimally deviating formulation by recursively finding the most constrained vehicle, reserving its optimal crossing time, and scheduling the crossing of the remaining vehicles using the previous schedule as constraints. This method allows to minimize the deviation between the overridden and desired controls, but may be computationally intensive. Indeed, one of the major difficulties of performing optimization in this context lies in the necessity to consider all the possible orderings of the vehicles. 

This problem is highly combinatorial; it has been shown that there exists up to $2^{n(n-1)/2}$ orderings for $n$ vehicles~\cite{Gregoire2014}. Moreover, it is generally ignored by most authors studying motion planning problems, who either use simple heuristics such as first-come, first-served~\cite{Dresner2004,Kim2014} or rely on exhaustive search~\cite{Simeon2002,Campos2015}. A possible method to handle the combinatorial explosion is to use pruning techniques such as \textit{branch-and-bound}, which avoid exploring branches of the decision tree that would provably yield suboptimal results. These methods are commonly used in mixed-integer linear (see, \textit{e.g.}, \cite{Peng2005,Muller2016} for applications to motion planning) or quadratic programming (see, \textit{e.g.}, \cite{Park2015}) problems, which combine continuous and discrete optimization. More general nonlinear methods have also been used in motion planning~\cite{Borrelli2006,Cafieri2014}, although their high computational difficulty generally requires linearization for effective resolution, as illustrated in~\cite{Murgovski2015a}. To the best of the authors' knowledge, branch-and-bound methods have never been applied to semi-autonomous driving.

This article significantly differs from references~\cite{Colombo2014,Campos2015,Ahn2015}. Instead of using a scheduling approach, we formulate the supervision problem as a Mixed Integer Quadratic Programming (MIQP) problem, which can handle various geometries with multiple collision points such as multi-lane intersections, merging lanes or roundabouts. Our formulation only requires to consider a small, finite planning horizon, while previous approaches~\cite{Colombo2014,Campos2015,Ahn2015} needed to schedule the crossing of all the considered vehicles. Furthermore, the MIQP formulation is highly flexible, allowing to take into account various constraints (\textit{e.g.}, maximal turning speed) and different cost functions. Finally, the resolution of MIQP can leverage highly-optimized solvers~\cite{gurobi}, allowing real-time implementations even for a relatively large number of vehicles.

This article expands the results presented in the conference paper~\cite{Altche2016b}; among the significant improvements made in this extension, we now give a more comprehensive model of our vision of semi-autonomous vehicles and adjust the modeling of the problem to handle bounded control errors. We provide a detailed discussion on how complex road geometries with multiply-intersecting paths can be handled, leading to a very versatile framework. Finally, we extend the theoretical results to continuous arrivals of vehicles, and provide possible ways for actual implementation as a roadside unit.

\section{Supervision problem}\label{sec:problem-statement}
We consider the problem of safely coordinating multiple semi-autonomous vehicles on the road, in order to prevent collisions and deadlock situations where no vehicle is able to move forward. Since vehicles are human-driven, a form of outside supervision is necessary to prevent undesirable situations. This section presents our formulation of a so-called Supervision problem generalizing the work of Colombo et al.~\cite{Colombo2014}; solving this problem yields a provably safe control, as close as possible to the original intentions of the drivers.

\subsection{Modeling}
\subsubsection{Supervision area}
We consider an isolated portion of a road infrastructure used by semi-autonomous vehicles, where some form of coordination is required to ensure vehicles safety. For instance, this could be a classical road intersection, a roundabout or an entry or weaving lane on a highway. We call this bounded portion of infrastructure the \textit{supervision area} and we assume that vehicles can travel safely outside of the collision area using only their ACC capacities. In a real-world setting, different critical portions of infrastructure which are far enough apart can be considered individually, but need to be treated jointly if traffic from one can influence another. \Cref{fig:supervision-area} shows examples of roads configurations and the corresponding possible choice for a supervision area. 

\begin{figure}
\centering
\begin{minipage}{0.5\columnwidth}
\subfloat[Crossroads]{\includegtikz[width=0.8\linewidth]{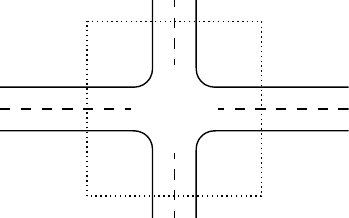}} \\
\subfloat[Roundabout]{\includegtikz[width=0.8\linewidth]{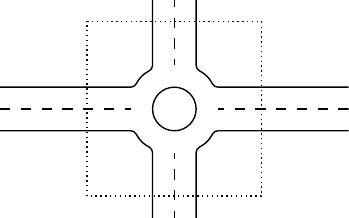}}
\end{minipage} 
\begin{minipage}{0.3\columnwidth}
\subfloat[Highway merging]{\includegtikz[width=\linewidth]{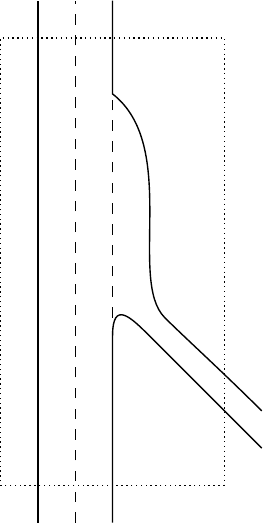}}
\end{minipage}
\caption{Examples of considered road configurations, and corresponding supervision areas (interior of the dotted rectangles).\label{fig:supervision-area}}
\end{figure}

In this article, we present an embodiment of a Supervisor working over a spatially static supervision area over time, that can be thought of as a dedicated computer on the infrastructure or in the cloud. Vehicles are assumed to establish a connection to the supervisor when they enter the supervision area (using, for instance, V2I communication), and maintain it until they exit this region. We denote by $\NN_t$ the set of vehicles currently inside the supervision area at a time $t$.

\subsubsection{Semi-autonomous vehicles}
We consider \textit{semi-autonomous} vehicles equipped with advanced driver assistance systems, many of which are already commercially available, and Vehicle to Infrastructure (V2I) communication capacities. In particular, vehicles are assumed to have advanced cruise control, automated braking and lane keeping assistance systems such that accelerating, braking and steering can be actuated by an on-board computer. Moreover, we suppose that vehicles have access to reliable cartographic data and are capable of precisely measuring their current position, orientation and velocity with reference to a unique global frame, for instance using GNSS and inertial navigation.

Since the vehicles are not assumed to have advanced environment-sensing capacities, for instance based on LIDAR data, they are not able to handle all situations and still require a human driver to safely navigate, for instance in the case of on-road obstacles or loss of GNSS signal. Moreover, lateral collisions or deadlock situations can happen due to human error, justifying the need for supervision.

\subsubsection{Parametrization}
In the remainder of this article, we only consider the two-dimensional kinematics and dynamics of the vehicles. We denote by $\boldmath E_i$ a bounding polygon for the shape of vehicle $i \in \NN_t$, and by $c_i$ the center of $\boldmath E_i$.

We assume that the geometry and lane markings of the roads inside the supervision area define a finite number of reference paths across this region, as exemplified in~\cref{fig:collision-region}. Due to the presence of a lane keeping assistance system, we assume that every vehicle is able to follow one of these reference paths with a small bounded lateral error. Noting $\gamma_i$ the reference path of a vehicle $i$, we assume that the distance of $c_i$ from $\gamma_i$ is bounded from above by $\xi_i \geq 0$. Moreover, we assume that $\gamma_i$ is at least $\mathcal C^2$-continuous, and that $\xi_i$ is small enough to ensure, for all $x \in \mathbb R^2$,
\begin{equation}
d(x,\gamma_i) \leq \xi_i \Rightarrow \exists !\ y \in \gamma_i, \ ||x-y|| = d(x, \gamma_i).
\end{equation}

This condition allows to use the curvilinear position of the point of $\gamma_i$ closest to $c_i$ to uniquely encode the longitudinal position of vehicle $i$ along $\gamma_i$. We denote by $s_i$ this curvilinear position, with the convention that $s_i = 0$ when the front bumper of $i$ first enters the supervision area and increases when $i$ goes forward; we let $s_i^{out}$ be the longitudinal position at which the rear bumper of $i$ fully exits the supervision area. \label{sec:coll-region}

\subsubsection{Vehicle dynamics}\label{sec:vehicle-dynamics}
In this article, we mostly focus on the longitudinal dynamics of the vehicles, and we let $x_i = \left( s_i, v_i \right)^T$ be the state of vehicle $i$, where $s_i$ and $v_i$ are respectively its longitudinal position and longitudinal speed. We assume that vehicles follow second-order integrator dynamics with a bounded longitudinal error, and that the control input $u_i$ corresponds to the longitudinal acceleration as:
\begin{equation}
	\dot x_i = A x_i + B u_i,
\end{equation}
where $A = \left( \begin{smallmatrix}0 & 1 \\ 0 & 0 \end{smallmatrix}\right)$ and $B = \left(\begin{smallmatrix}0 \\ 1\end{smallmatrix}\right)$. Since we mostly consider situations with conflicting vehicles, we assume that human drivers maintain a relatively low speed (compared to the curvature of their path), which allows neglecting lateral dynamics and slip~\cite{Polack2017}.
 
To account for speed limitations on the vehicles, each vehicle $i$ is supposed to have a bounded non-negative velocity, so that $v_i \in [0, \overline v_i]$ (with $\overline v_i > 0$) at all times. Moreover, we assume that the acceleration $u_i$ of each vehicle is bounded as $u_i \in [\underline u_i, \overline u_i]$, with $\underline u_i < 0 < \overline u_i$. These bounds can differ between vehicles, thus allowing heterogeneous vehicle performance. At a given time $t$, we let $\bU_t = \prod_{i \in \NN_t} [\underline u_i, \overline u_i]$ be the set of admissible accelerations for the vehicles of $\NN_t$. We denote bt boldface $\mathbf x = (x_i)_{i \in \NN_t}$ and $\bu = (u_i)_{i \in \NN_t}$ the state and control for the system of vehicles.

In what follows, we let $v_{max} > 0$ be a global upper bound for $\overline{v}_i$, $u_a > 0$ a lower bound for $\overline u_i$ and $u_b < 0$ an upper bound for $\underline{u}_i$ such that for all $t \geq t_\kappa$ and all $i \in \NN_t$, $\overline v_i \leq v_{max}$ and $\underline u_i \leq u_{b} < 0 < u_a \leq \overline u_i$. Therefore, all vehicles are capable of braking with $u_b$ and accelerating with $u_a$; finally, we let $u_{max}$ be a global upper bound for $\overline u_i$. 

\subsubsection{Collision regions}
Finally, we assume that the angle between the orientation of vehicle $i$ and the tangent to $\gamma_i$ at its point closest to $c_i$ is also bounded. With these hypotheses, for any pair of vehicles $(i,j)$, we can compute the bounded set $\mathcal C_{ij}$ of curvilinear positions $(s_i, s_j) \subset [0,s_i^{out}] \times [0,s_j^{out}]$ for which a collision could happen between $i$ and $j$. Note that these sets are ``inflated'' to take into account the bounded control errors. We call $\mathcal C_{ij}$ the \textit{collision region} between $i$ and $j$; \cref{fig:collision-region} shows examples of paths and corresponding computed collision regions for different driving situations. Note that collision regions can be empty or have one or multiple connected components. If $\mathcal C_{ij} \neq \emptyset$, we say that vehicles $i$ and $j$ are \textit{conflicting}; when $\mathcal C_{ij}$ has multiple connected components, we denote by $\mathcal C^p_{ij}$ its $p$-th component, using the convention  $\mathcal C^p_{ij} = \mathcal C^p_{ji}$.

\begin{figure}
\subfloat[Simple orthogonal intersection situation with example vehicle shapes\label{fig:collision-region-a}]{
\includegtikz[width=0.45\linewidth,height=3cm]{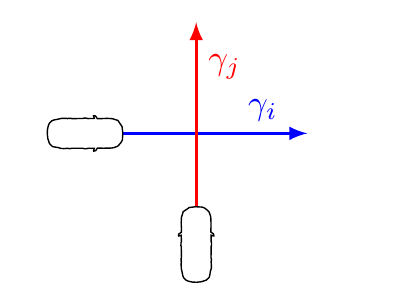}
\includegtikz[width=0.45\linewidth,height=3cm]{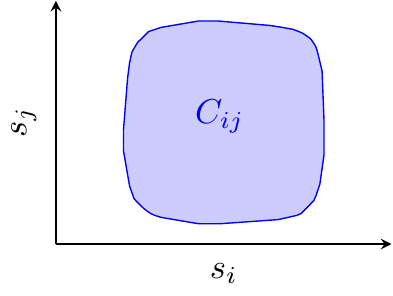}}

\subfloat[Roundabout situation with multiple connected components in $\mathcal C_{ij}$\label{fig:collision-region-b}]{
\includegtikz[width=0.45\linewidth,height=3cm]{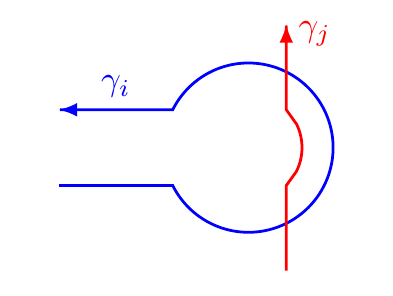}
\includegtikz[width=0.45\linewidth,height=3cm]{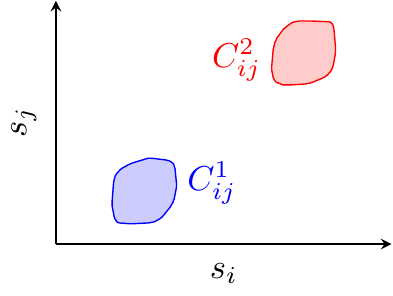}}

\subfloat[Highway merging situation\label{fig:collision-region-c}]{
\includegtikz[width=0.45\linewidth,height=3cm]{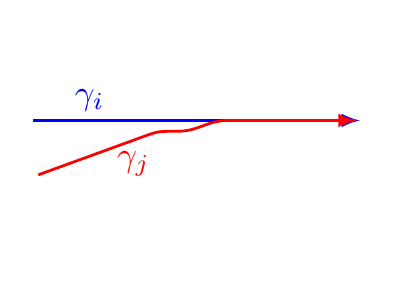}
\includegtikz[width=0.45\linewidth,height=3cm]{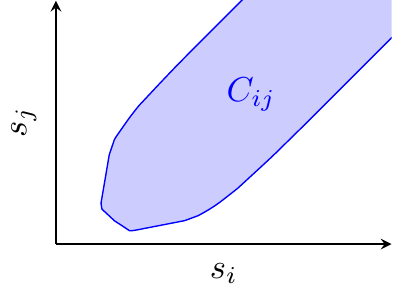}}
\caption{Examples of paths (left) and corresponding collision regions (right) for vehicles with the polygonal shape shown in~\cref{fig:collision-region-a}. \label{fig:collision-region}}
\end{figure}

\subsubsection{No-stop regions}
To prevent creating deadlock situations, vehicles are not allowed to stop when doing so would block traffic in other directions. To this extent, we define a \textit{no-stop region} (see~\cref{fig:nostop-accel-region}) $\mathcal D_i$ for each vehicle $i \in \NN_t$ as the smallest interval $\mathcal D_i = [\underline s_i^\perp, \overline s_i^\perp]$ containing all $\min \left( \Pi_{s_i} \mathcal C^p_{ij} \right)$ for all $t' \geq t$, $j \in \NN_{t'}$ and all $p$ such that $(0,0) \notin \mathcal C^p_{ij}$; in this formula, $\Pi_{s_i}$ is the projection operator on the first coordinate. The no-stop region corresponds to the part of the supervision area where a vehicle may have to yield to another; if $\mathcal C^p_{ij}$ contains $(0,0)$, then either $i$ or $j$ enters the supervision area behind the other, in which case the relative ordering of the vehicles is given and the $\mathcal C^p_{ij}$ does not count in $\mathcal D_i$.

Note that, although this definition theoretically requires knowledge of all future vehicles, $\mathcal D_i$ can be computed off-line as a finite intersection of intervals provided that there only exists a finite number of possible paths inside the supervision area. In what follows, we let $v_{min} > 0$ be a minimum allowed speed for any vehicle inside its no-stop region, and we assume that $v_{min} \leq \overline v_i$ for all vehicles.

For a no-stop region $\mathcal D_i$, we define the corresponding acceleration region $\mathcal A_i = [s_i^{acc}, \underline s_i^\perp]$ such that, if vehicle $i$ is stopped at $s_i^{acc}$, it can reach a speed $v_{min}$ before reaching $\underline s_i^\perp$. More specifically, we require that $0 \leq s_i^{acc} \leq \underline s_i^\perp - \frac{{v_{min}}^2}{2 u_a}$ for all $i$. Inside the acceleration region, vehicles are only allowed to accelerate; this condition prevents vehicles from stopping right before the entrance of the no-stop region, leaving them unable to proceed forward due to the minimum speed requirement. \Cref{fig:nostop-accel-region} illustrates an example of the no-stop regions and the corresponding acceleration regions. \label{sec:nostop-region}

\begin{figure}
\includegtikz[width=\columnwidth,height=4cm]{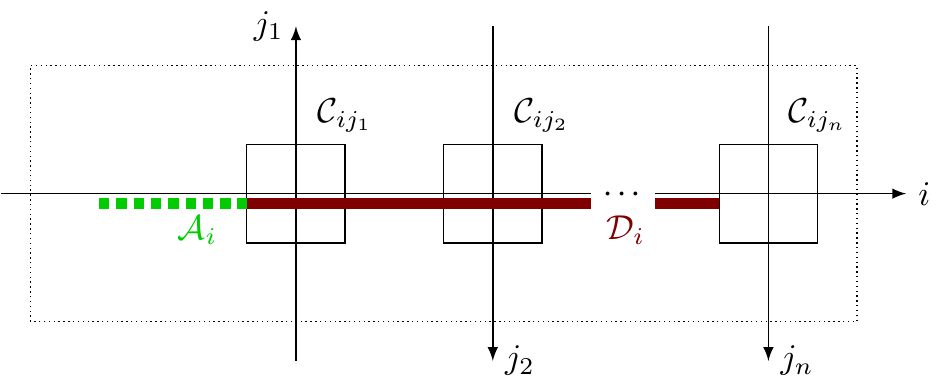}
\caption{Illustration of the no-stop region $\mathcal D_i$ and acceleration region $\mathcal A_i$ inside the supervision area (dotted rectangle).\label{fig:nostop-accel-region}}
\end{figure}

\subsubsection{Time discretization}
Drivers continuously change the control input of their vehicle; however, due to computational and communication constraints, it is impractical to handle functions of a continuous variable. In the remainder of this article, we choose a constant time step duration $\tau > 0$, and we assume that all vehicles use piecewise-constant controls with step $\tau$, typically \SI{0.5}{\second}. To simplify the formulation, we further assume that vehicles update their control simultaneously at times $t_\kappa = \kappa \tau$ for $\kappa \in \mathbb N$, and we denote by $\mathcal U_\tau(t_\kappa)$ the set of piecewise-constant admissible controls for the vehicles of $\NN_{t_\kappa}$. By definition, for all $t \geq t_\kappa$ and all $\bu \in \mathcal U_\tau(t_\kappa)$, $\bu(t) \in \bU_{t_\kappa}$.

\subsection{Problem statement}
Before presenting the so-called \textit{supervision problem}, we first define the safety criterion for the vehicles inside the supervision area at a given time.
\begin{defn}[Safe state]\label{def:safe-state}
	We say that the supervision area is in a safe state $\mathbf x^\kappa$ at time $t_\kappa$ if there exists an admissible piecewise-constant control $\bu \in \mathcal U_\tau(t_\kappa)$ defined over $[t_\kappa, +\infty[$ such that, under this control and starting from $\mathbf x^\kappa$, for all $t \geq t_\kappa$ and all $i, j \in \NN_{t_\kappa}$, $(s_i(t), s_j(t)) \notin \mathcal C_{ij}$. Such a control is said to be a \textit{safe control}.
\end{defn}

With this definition, the supervision area is in a safe state when all the vehicles inside this area can apply a dynamically admissible, infinite horizon control without a risk of collision. This safety condition corresponds to a contraposition of the notion of ``{inevitable collision state}'' proposed by Fraichard et al.~\cite{Fraichard2003}. In what follows, we denote by $\mathcal U^{safe}_\tau(t_\kappa)$ the set of safe and dynamically admissible piecewise-constant controls for the vehicles in $\NN_{t_\kappa}$; by definition, a control $\bu \in \mathcal U^{safe}_\tau(t_\kappa)$ is a piecewise-constant function from $[t_\kappa,+\infty[$ to $\bU_{t_\kappa}$. We now define the safety condition for vehicles entering the supervision area.
\begin{defn}[Safe entry]\label{def:safe-entry}
	Consider a safe state $\mathbf x^\kappa$ at time $t_\kappa$ and let $t_1 > t_\kappa$ be the first time at which a new vehicle enters the supervision area. We say that the vehicles of $\NN_{t_1} \setminus \NN_{t_\kappa}$ safely enter the supervision area with a margin $\tau$ if $t_1 \geq t_\kappa + \tau$, or if any safe control $\bu \in \mathcal U^{safe}_\tau(t_\kappa)$:
	\begin{itemize}
		\item keeps the system of the vehicles of $\NN_{t_1}$ safe at time $t_\kappa + \tau$ and
		\item remains safe over $[t_\kappa, +\infty[$ for the vehicles of $\NN_{t_\kappa}$,
	\end{itemize} regardless of the control applied by the vehicles of $\NN_{t_1} \setminus \NN_{t_\kappa}$ over $[t_1, t_1+\tau]$.
\end{defn}

This definition ensures that a safe control computed for the vehicles of $\NN_{t_\kappa}$ remains safe after new vehicles enter, \textit{i.e.} the entry of new vehicles does not invalidate previously safe controls. Moreover, we assume that we can safely exclude vehicles departing the supervision area from the safety verification problem, \textit{i.e.} that drivers are able to safely follow the previously departed vehicles without supervision. We will show in \Cref{sec:infinite-rh} that these hypotheses allow discrete-time supervision with continuous vehicle arrival.

In the remainder of this article, we consider a centralized supervisor working in discrete time steps of duration $\tau$, and we assume that new vehicles always enter safely with a margin $\tau$. At the beginning of each time step $\kappa$, the supervisor receives an information about the desired longitudinal control of each vehicle for the next time step, denoted by $u_{i,des}^\kappa$. The collection of these desired controls for the vehicles of $\NN_{t_\kappa}$ defines a constant desired system control $\bu_{des}^\kappa$ defined over $[t_\kappa, t_\kappa + \tau [$.

This control may, or may not, lead the system of vehicles into an unsafe state. The supervisor is tasked with preventing the system from entering an unsafe state, by overriding the desired control if necessary. To remain compatible with human drivers, it is desirable that the supervisor has several properties, namely being \textit{least restrictive} and \textit{minimally deviating}. Letting $\mathcal U^{safe}_{\tau,\kappa}(t_\kappa)$ be the restriction of the functions of $\mathcal U^{safe}_{\tau}(t_\kappa)$ to $[t_\kappa, t_\kappa+\tau[$, we define the \textit{least restrictive} supervision problem:
\begin{defn}[Least restrictive supervision]
	Consider a safe state $\mathbf x^\kappa$ at time $t_\kappa = \kappa \tau$, a desired system control $\bu_{des}^\kappa$ and assume that all new vehicles enter the supervision area safely with a margin $\tau$. The least restrictive \textit{supervision problem} ($SP$) is that of finding a control $\bu_{safe}^\kappa \in \mathcal U^{safe}_{\tau,\kappa}(t_\kappa)$ such that 	$\bu_{safe}^\kappa = \bu_{des}^\kappa$ if $\bu_{des}^\kappa \in \mathcal U^{safe}_{\tau,\kappa}(t_\kappa)$.
\end{defn} 

Note that this definition corresponds to that of~\cite{Colombo2014} in our generalized setting. Such a supervisor is \textit{least restrictive} because overriding only occurs if the initially requested control would lead the vehicles in an unsafe state. However, it is also desirable that the control used for overriding is chosen close to the drivers' desired control. Extending the work in~\cite{Campos2015}, we define the \textit{minimally deviating} supervision problem as follows:
\begin{defn}[Minimally deviating supervision]\label{def:minimally-deviating}
	Consider a safe state $\mathbf x^\kappa$ at time $t_\kappa = \kappa \tau$, a desired system control $\bu_{des}^\kappa$ and assume that all new vehicles enter the supervision area safely with a margin $\tau$. The minimally deviating \textit{supervision problem} ($SP^*$) is that of finding a constant control ${\bu^*}_{safe}^\kappa$ such that:
	\begin{align}
	\mathbf {\bu^*}_{safe}^\kappa = \argmin_{\bu\, \in\, \mathcal U^{safe}_{\tau,\kappa}(t_\kappa)} ||\bu^\kappa - \bu_{des}^\kappa||
	\end{align}
	where $||\cdot||$ is a norm defined over $\bU_{t_\kappa}$.
\end{defn} 
Note that, from this definition, any solution to $SP^*$ is a solution to $SP$.

This concept of minimally deviating supervision follows a different fail-safety paradigm that could be found in, \textit{e.g.}, rail transportation where all trains in an area should perform an emergency braking when an incident occurs. The reasoning behind \cref{def:minimally-deviating} is that, to improve efficiency without sacrificing safety, intervention is only performed on vehicles which are actually at risk, and does not necessarily result in a full stop. However, at individual vehicle level, the safe overriding control ${u^*}_{safe}^\kappa$ may differ greatly from the driver's input, \textit{e.g.} braking instead of accelerating.

\section{Infinite Horizon Formulation of the Supervision Problem}\label{sec:miqp-supervised}
In this section, we present an extension of the work in~\cite{Altche2016b} allowing to reformulate the generalized minimally deviating supervision problem using mixed-integer quadratic programming (MIQP) in~\Cref{sec:miqp}. As the supervisor works in discrete time steps of duration $\tau$, we consider the beginning of a step $\kappa$, corresponding to a time $t_\kappa = \kappa \tau$ and formulate an infinite-horizon MIQP problem. Assuming the initial state is safe, we will show in \Cref{sec:miqp-obj} that this formulation can be used to find a minimally deviating safe control for the vehicles in $\NN_{t_\kappa}$. We will show in \Cref{sec:infinite-rh} that, if the vehicles of $\NN_{t_\kappa}$ follow the corresponding control, our formulation expressed at $t_{\kappa+1} = (\kappa+1)\tau$ remains feasible for the vehicles of $\NN_{t_{\kappa+1}}$, provided that all new vehicles enter safely with a margin $\tau$. These properties ensure that our infinite horizon MIQP formulation can be solved in a receding horizon fashion, to ensure safety for all future vehicles.

\subsection{Model variables and constraints}\label{sec:miqp}
In what follows, we present the variables and constraints used in our model. Unless specified otherwise, these constraints are enforced at all time steps $k \geq \kappa$, and for all vehicles of $\NN_{t_\kappa}$.

\subsubsection{Vehicle dynamics}
When they evolve inside the supervision area, vehicles use a piecewise-constant control, which is updated every $\tau$ seconds. For a vehicle $i \in \NN_{t_\kappa}$ at step $k$, we introduce the variables $s_i^k$, $v_i^k \in [0, \overline v_i]$ and $u_i^k \in [\underline u_i, \overline u_i ]$, respectively denoting its curvilinear position and longitudinal speed at $t_k$, and longitudinal acceleration over $[t_k, t_k+\tau[$. The following constraints enforce vehicle dynamics:
\begin{align}
	s_i^{k+1} - s_i^{k} = &\ \frac 1 2 \left( v_i^k + v_i^{k+1} \right)\tau \label{eq:cstr-first} \\
	v_i^{k+1} - v_i^{k} = &\ u_i^k\tau \label{eq:cstr-dyn}
\end{align}

\subsubsection{Logical constraints}
In~\cite{Altche2016}, we showed that it is possible to enforce logical constraints on continuous and integer variables with linear inequalities using a ``big-M'' formulation. More specifically, if $b$ is a binary variable and $x$ a continuous or integer variable bounded so that $|x| < M$, then the logical constraint: $(b = 0 \Rightarrow x \leq a)$ is equivalent to the linear inequality constraint $x \leq a + bM$. This method can be used to define indicator binary variables for a given semi-infinite interval: for a continuous variable $x$ and a constant $a \in \mathbb R$, we denote by $b = \chi_{[a,+\infty[}(x)$ the constraints $(b = 0 \Rightarrow x \leq a) \wedge (b = 1 \Rightarrow x \geq a)$, where $\wedge$ denotes the binary conjunction; we use $\neg$ to denote the binary negation.

\subsubsection{Collision avoidance}
As presented in \Cref{sec:coll-region}, the \textit{collision region} between two vehicles $i, j \in \NN_{t_\kappa}$, $\mathcal C_{ij}$, can be computed off-line. As it was already presented in~\cite{Altche2016}, it is possible to compute a minimal bounding convex polygon for each connected component $\mathcal C^p_{ij}$ of $\mathcal C_{ij}$. A good compromise between accuracy and complexity is to use a bounding hexagon with edges either parallel to the $s_i = 0$, $s_i = s_j$ or $s_j = 0$ lines; such a polygon is uniquely defined by six parameters, as shown in \cref{fig:bounding-hexagon}. 

\begin{figure}\centering
	\includegtikz[width=0.55\linewidth,height=.4\linewidth]{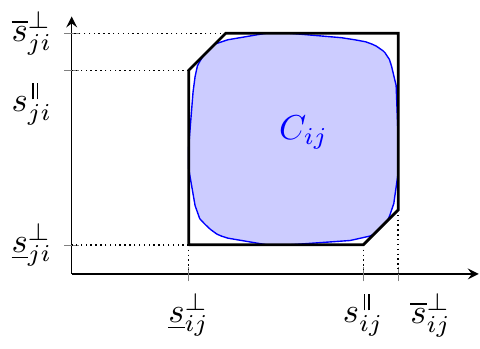}
	\caption{Minimum bounding hexagon for the collision region presented in \cref{fig:collision-region-a}.\label{fig:bounding-hexagon}}
\end{figure}

To ensure that vehicles do not enter any of the collision regions, we introduce a set of binary variables to encode the discrete decisions arising from the choice of an ordering of vehicles, as presented in~\cite{Altche2016}. For all conflicting vehicles $i,j \in \NN_{t_\kappa}$, we let $\pi_{ij}^p = 1$ if vehicle $i$ passes the $p$-th collision region before $j$, and $\pi_{ij}^p = 0$ otherwise; moreover, we introduce the binary indicator variables for all $k \geq \kappa$:
\begin{align}
\varepsilon_{ij,p}^\parallel(k) = &\ \chi_{[s_{ij,p}^\parallel,+\infty[}(s_i^k), \\
\varepsilon_{ij,p}^\perp(k) = &\ \chi_{[\overline s_{ij,p}^\perp,+\infty[}(s_i^k).
\end{align}

We enforce the collision avoidance constraints for all conflicting vehicles $i,j \in \NN_{t_\kappa}$ and $k \geq \kappa$ as:
\begin{align}
	\left( \pi_{ij}^p \wedge \neg\, \varepsilon_{ij,p}^\parallel(k) \right) & \Rightarrow s_j^{k+1} \leq \underline s_{ji,p}^\perp \label{eq:safety-1} \\
	\left( \pi_{ij}^p \wedge \varepsilon_{ij,p}^\parallel(k) \wedge \neg\, \varepsilon_{ij,p}^\perp(k) \right) &  \Rightarrow  s_i^{k+1}  \geq s_j^{k+1} + d_{ij,p} \label{eq:safety-2} \\
	\left( \pi_{ij}^p \wedge \varepsilon_{ij,p}^\parallel(k) \wedge \neg\, \varepsilon_{ij,p}^\perp(k) \right) & \Rightarrow \nonumber \\
	 s_i^{k+1} \geq s_j^{k+1} + & d_{ij,p}  + \frac \tau 2 \left(v_j^{k+1} - v_i^{k+1} \right) \label{eq:safety-3}
\end{align}
where $d_{ij,p} = s_{ij,p}^\parallel - \underline s_{ji,p}^\perp$. Constraint~\eqref{eq:safety-1} corresponds to ``crossing situations'', where a vehicle has to wait for another to pass; constraints~\eqref{eq:safety-2} and~\eqref{eq:safety-3} correspond to ``following situations'', where a vehicle needs to maintain a certain longitudinal distance from another.

Note that constraints \eqref{eq:safety-1} to \eqref{eq:safety-3} use the values of the indicator variables at step $k$ to force the positions of the vehicles at step $k+1$ in order to avoid a ``corner cutting'' phenomenon; the additional constraint \eqref{eq:safety-3} prevents collisions between two time steps. These constraints are very slightly stronger than that of collision avoidance, \textit{i.e.} for all $t\geq t_\kappa$, $(s_i(t), s_j(t)) \notin \mathcal C_{ij}$. Consequently, the results in the rest of this article are to be understood replacing the exact collision avoidance constraints in \cref{def:safe-state} by conditions \eqref{eq:safety-1}-\eqref{eq:safety-3}.

Finally, to ensure the consistency of the formulation, we add the mutual exclusion constraint for all conflicting vehicles $i, j \in \NN_{t_\kappa}$:
\begin{equation}
	\pi_{ij}^p + \pi_{ji}^p = 1.
\end{equation}

\subsubsection{Deadlock avoidance}
As described in \Cref{sec:nostop-region}, we require all vehicles to maintain a minimum speed inside their no-stop region $\mathcal D_i = [\underline s_i^\perp, \overline s_i^\perp]$. This requirement is enforced by defining  additional binary variables, for all $i \in \NN_{t_\kappa}$ and all $k \geq \kappa$:
\begin{align}
	\zeta_i^{acc}(k) =& \chi_{[\underline s_i^{acc}, +\infty[} (s_i^k) \\
	\zeta_i^{in}(k) =& \chi_{[\underline s_i^\perp, +\infty[} (s_i^k) \\
	\zeta_i^{out}(k) =& \chi_{[\overline s_i^\perp, +\infty[} (s_i^k) \\
	\eta_i(k) =& \chi_{[v_{min} - u_a \tau, +\infty[} (v_i^k)
\end{align}
and using the constraints:
\begin{align}
	\left( \zeta_i^{acc}(k) \wedge \neg\, \zeta_i^{in}(k) \wedge \neg\, \eta_i(k)  \right) \Rightarrow&\ v_i^{k+1} \geq v_i^k + u_a \tau,\label{eq:cstr-accel} \\
	\left( \zeta_i^{in}(k) \wedge \neg\, \zeta_i^{out}(k) \right) \Rightarrow&\ v_i^k \geq v_{min}.\label{eq:cstr-deadlock}
\end{align}
As long as the acceleration regions $\mathcal A_i$ are large enough, constraint~\eqref{eq:cstr-accel} prevents vehicles from remaining blocked due to the minimum speed requirement~\eqref{eq:cstr-deadlock}. We will show in the next section that these conditions effectively prevent deadlocks for all future times.

\subsubsection{Initial conditions}
The supervision problem is used in a receding horizon fashion, and we consider that the state of each vehicle of $\NN_{t_\kappa}$ at time $t_\kappa$ is known before solving the problem. Therefore, we use the following initial condition for all $i \in \NN_{t_\kappa}$:
\begin{equation}
	\left( s_i^\kappa, v_i^\kappa \right) = \left( s_i(t_\kappa), v_i(t_\kappa) \right)\label{eq:cstr-last} 
\end{equation}

\subsection{Objective function}\label{sec:miqp-obj}
Any piecewise-constant control verifying constraints \eqref{eq:cstr-first} to \eqref{eq:cstr-last} for all $k \geq \kappa$ is dynamically admissible and prevents collisions for all future times, and is therefore in $\mathcal U^{safe}_{\tau,\kappa}(t_\kappa)$. To remain compatible with human driving, we now formulate an objective function allowing to find a least restrictive and minimally deviating control given a desired control $\bu_{des}^\kappa = (u_{i,des}^\kappa)_{i\in \NN_{t_\kappa}}$. In what follows, we let $(w_i^\kappa)_{i \in \NN_{t_\kappa}}$ be a set of strictly positive weights, $\mathbf X$ be the tuple of all the problem variables, and we define:
\begin{equation}
	J^\kappa(\mathbf X) = \sum_{i \in \NN_{t_\kappa}} w_i^\kappa \left( u_i^\kappa - u_{i,des}^\kappa \right)^2.
\end{equation}
Noting $\pi_{\bu^\kappa}$ the projection operator such that $\pi_{\bu^\kappa}(\mathbf X) = \bu^\kappa$, we deduce the following theorem:

\begin{theorem}\label{thm:ih-miqp}
The solution of the optimization problem:
\begin{align}
	\bu^* =&\ \pi_{\bu^\kappa} \argmin_{\mathbf X} J^\kappa(\mathbf X) \label{eq:i-miqp}\tag{IH-SP} \\
	\text{subj. to }&\ \forall k \geq \kappa,\ \eqref{eq:cstr-first}-\eqref{eq:cstr-last} \nonumber
\end{align}
is a solution to the minimally deviating supervision problem $SP^*$ at time $t_\kappa$, for the norm associated with ${J^\kappa \circ \pi_{\bu^\kappa}}$.
\end{theorem}

Note that the weighting terms $w_i^\kappa$ allow distinguishing between different types of agents, for instance to prioritize emergency services or high-occupancy vehicles. More complex cost functions can also be used, for instance to penalize a forced acceleration more than a forced braking. 

\subsection{Receding horizon properties}\label{sec:infinite-rh}
We now assume that there exists a solution to~\ih{} at time $t_\kappa$, that the vehicles of $\NN_{t_\kappa}$ follow this solution control over $[t_\kappa, t_\kappa+\tau]$, and that the vehicles of $\NN_{t_{\kappa+1}}$ enter safely with a margin $\tau$. From \Cref{def:safe-state,def:safe-entry}, we have the following theorem:
\begin{theorem}[Recursive feasibility]\label{thm:feasible}
	Let $\tau > 0$, $\kappa \geq 0$, $t_\kappa = \kappa \tau$ and  $t_{\kappa+1} = t_\kappa + \tau$. Assume that:
	\begin{itemize}
		\item there exists a solution to~\ih{} at time $t_\kappa$ for the vehicles of $\NN_{t_\kappa}$,
		\item the vehicles of $\NN_{t_\kappa}$ follow this solution control over $[t_\kappa, t_{\kappa+1}]$,
		\item the vehicles of $\NN_{t_{\kappa+1}} \setminus \NN_{t_\kappa}$ enter safely with a margin $\tau$.
	\end{itemize}
	Then there exists a solution to~\ih{} at time $t_{\kappa+1}$ for the vehicles of $\NN_{t_{\kappa+1}}$.
\end{theorem}
\begin{proof}From \cref{def:safe-state,def:safe-entry}, and using \cref{thm:ih-miqp}, we know that the first two hypotheses guarantee that the vehicles in $\NN_{t_\kappa}$ are in a safe state at time $t_{\kappa+1}$. Moreover, the third hypothesis ensures that the vehicles in $\NN_{t_{\kappa+1}}$ also are in a safe state at $t_{\kappa+1}$ regardless of the control applied by the vehicles of $\NN_{t_{\kappa+1}} \setminus \NN_{t_{\kappa}}$ up to time $t_{\kappa+1}$. By \cref{def:safe-state}, there exists a feasible solution to~\ih{} thus proving the theorem.
\end{proof}

We now state that the~\ih{} formulation effectively prevents the apparition of deadlocks%
\ifapp
The proof of this theorem can be found in \Cref{app:proof-deadlock}. 
\else
; this theorem is demonstrated in \cite{Altche2017}.
\fi 
\begin{theorem}[Deadlock avoidance]\label{thm:deadlock}
	Let $\kappa \geq 0$ and assume that, for all $\kappa \leq k \leq \kappa_0$, the conditions of \cref{thm:feasible} remain satisfied at time $t_k$. There exists a feasible solution of \ih{} at time $t_{\kappa_0}$ in which all the vehicles in $\NN_{t_{\kappa_0}}$ exit the supervision area in finite time.
\end{theorem}

Note that \cref{thm:deadlock} only ensures that, at all times, there exists a solution where all the vehicles inside the supervision at this particular time eventually exit. However, there is no guarantee that such a solution will actually be selected, for instance if one driver wishes to stop although there is no other vehicle. There is also no fairness guarantee, \textit{i.e.} it is possible that one vehicle is forced to remain stopped for an arbitrarily long time, for instance if there is a very heavy traffic coming from another direction. Future developments will focus devising more complex objectives function to take traffic efficiency and fairness into account.

\subsection{Multiple paths choices}
The above formulation assumes that the path of each vehicle is known in advance. However, this may not be realistic in the context of semi-autonomous cars where drivers can decide to change paths, for instance to avoid an obstacle on the road or use another itinerary. Using additional variables to indicate the path to which a vehicle is assigned, our formulation can be extended to handle multiple possible paths for each vehicle. Due to length limitations, this extension will be detailed in future work.

\section{Finite Horizon Formulation}\label{sec:finite-horizon}
In \Cref{sec:infinite-rh}, we presented an infinite horizon formulation to solve the minimally deviating supervision problem. However, due to the infinite number of variables, this formulation is not suitable for practical resolution. In this section, we derive an equivalent finite horizon formulation that can be implemented and solved using standard numerical techniques.

In what follows, we let $K \geq 1$ and we denote by \fh{K} the restriction of~\ih{} at time $t_\kappa$ to the variables at steps $k$ with $\kappa \leq k \leq \kappa+K$, and we only consider the constraints~\eqref{eq:cstr-first} to~\eqref{eq:cstr-last} up to step $\kappa+K$. The objective function is unchanged. A solution to \fh{K} at time $t_\kappa$ allows to compute a control preventing collisions up to time $t_\kappa + K\tau$; however, due to the dynamics of the vehicles, the state reached at $t_\kappa + K\tau$ may not be safe. Since \fh{K} only has a subset of the constraints of \ih{}, we can formulate the following proposition:

\begin{proposition}\label{prop:ih-to-fh}
Let $K \geq 1$ and let $\mathbf X$ be a solution of \ih{} at step $\kappa$. The restriction of $\mathbf X$ to the first $K+1$ time steps is a feasible solution to \fh{K}.
\end{proposition}

Using the global bounds $u_a$, $u_b$, $u_{max}$ and $v_{max}$ defined in \cref{sec:vehicle-dynamics}, we will now prove a reciprocal implication to \cref{prop:ih-to-fh}: if $K$ is chosen large enough, any solution of \fh{K} can be used to construct a solution of \ih{}.

As presented in~\cite{Altche2016b}, the key idea of the proof lies in the choice of a planning horizon long enough to allow any vehicle to fully stop. The structure of the demonstration is as follows: \cref{lem:canstop} gives a lower bound on the time horizon to allow a single isolated vehicle to stop using discrete dynamics, although with a potential risk of rear-end collisions from following vehicles. In \cref{prop:multcanstop}, we give a slightly higher bound on the time horizon ensuring that all vehicles in a line can all safely stop without rear-end collisions. Finally, in \cref{prop:recursive} we give a bound on $K$ ensuring the recursive feasibility of \fh{K}; this allows formulating \cref{thm:equivalence}, stating the equivalence of \fh{K} and \ih{}.
\ifapp In this section, we only present sketches of proofs for each result; detailed demonstrations can be found in \Cref{app:proofs-fh}.
\else Detailed demonstrations can be found in \cite{Altche2017}.
\fi

\begin{lemma}
\label{lem:canstop}At a time $t_\kappa$, consider a horizon $T = K \tau$ with $T \geq \frac{v_{max}}{|u_{b}|} + \tau$. Let $i \in \NN_{t_\kappa}$ be a vehicle for which there exists a piecewise-constant control $(u_i^k)_{\kappa \leq k < \kappa+K}$ such that, for all $\kappa \leq k < \kappa+K$, $u_i^k \in [\underline u_i, \overline u_i]$, corresponding to a dynamically feasible trajectory $s_i(t)$ over $[t_\kappa, t_\kappa + T + \tau]$.

\noindent There exists a discrete control $(\tilde u_i^k)_{\kappa \leq k \leq \kappa+K}$ such that for all $\kappa \leq k \leq \kappa+K$, $u_i^k \in [\underline u_i, \overline u_i]$ and $\tilde u_i^{\kappa} = u_i^{\kappa}$, and for which the corresponding dynamically feasible trajectory $t \mapsto \tilde x_i(t) = (\tilde s_i(t), \tilde v_i(t))$ verifies $\tilde s_i(t_\kappa + T + \tau) \leq s_i(t_1 + T)$ and $\tilde v_i = 0$ over $[t_\kappa + T, t_\kappa + T + \tau]$.
\end{lemma}
\begin{proof}[Sketch of proof]
$\frac{v_{max}}{|u_b|}$ is an upper bound on the required time for any vehicle to stop by applying a control $u_b$, which by definition is dynamically feasible. The additional $\tau$ accounts for the fact that we require $\tilde u_i^\kappa = u_i^\kappa$ at the first time step.
\end{proof}

In the following proposition and noting $\ceil{\cdot}$ the ceiling function, we prove a bound ensuring that a line of vehicles can safely stop before the leader reaches its final computed position at the end of the time horizon, without risk of rear-end collisions:
\begin{proposition}
\label{prop:multcanstop}At a time $t_\kappa$, suppose that $p$ vehicles of $\NN_{t_\kappa}$ (denoted by $1, \dots, p$ from rear to front) are following one another. Consider a horizon $T_{stop} = K_{stop}\tau \geq \frac{v_{max}}{|u_b|} + (p-1) \left(1+ \ceil{\frac{u_{max}}{|u_b|}}\right) \tau + \tau$, and assume that every vehicle $i \in \{1,\dots,p\}$ has a \textit{safe} discrete control $(u_i^k)_{\kappa \leq k < \kappa+K}$ such that, for all $\kappa \leq k < \kappa+K$, $u_i^k \in [\underline u_i, \overline u_i]$. We let $t \mapsto x_i(t)$ be the trajectory over $[t_\kappa, t_\kappa + T]$ for vehicle $i$ under control $(u_i^k)$.

\noindent For all $i \in \{1,\dots,p\}$, there exists a \textit{safe} discrete control $(\hat u_i^k)_{\kappa \leq k \leq \kappa+K}$ such that for all $\kappa \leq k \leq \kappa+K$, $u_i^k \in [\underline u_i, \overline u_i]$, $\hat u_i^{\kappa} = u_i^{\kappa}$ and for which the corresponding dynamically feasible and safe trajectory $t \mapsto \hat x_i(t) = (\hat s_i(t), \hat v_i(t))$ verifies $\hat s_i(t_\kappa + T + \tau) \leq s_i(t_1 + T)$ and $\hat v_i = 0$ over $[t_\kappa + T, t_\kappa + T + \tau]$.
\end{proposition}
\begin{proof}[Sketch of proof]
The worst case that needs to be taken into account corresponds to a situation where the initial states of the vehicles require each of them to accelerate in order to avoid a rear-end collision from the vehicle behind. This rather extreme situation happens when a vehicle goes faster than the one it is following, and the two are too close to allow a safe deceleration. In this case, the rearmost vehicle can always brake with the control from \cref{lem:canstop}, until it decelerates below the speed of the vehicle in front of it. The second rearmost vehicle can then decelerate, then the third and up to the front-most vehicle. The term $\left(1+ \ceil{\frac{u_{max}}{|u_b|}}\right) \tau$ arises from the piecewise-constant control hypothesis, and vanishes as $\tau$ goes to $0$. Note that the condition $K_{stop}\tau \geq \frac{v_{max}}{|u_b|} + \frac{v_{max}}{u_a} + \tau$ also provides the same guarantees; depending on the value of $p$, this second bound might be more efficient.
\end{proof}

\begin{remark}
	The bound from~\cref{prop:multcanstop} depends on the number of vehicles in a line, and can become quite high when $p$ is large. It can be proven that the condition $K_{stop}\tau \geq \frac{v_{max}}{|u_b|} + \frac{v_{max}}{u_a} + \tau$ also provides the same guarantees; depending on the value of $p$, this second bound might be more efficient.
\end{remark}

We can now prove the recursive feasibility of \fh{K} for a large enough $K$, as follows:
\begin{proposition}\label{prop:recursive}
Consider a time $t_\kappa$, and assume that at most $p$ vehicles are following one another at all times $t \geq t_{\kappa}$. We set $d = \max_{t \geq t_\kappa, i \in \NN_t} \left( \overline s_i^\perp - s_i^{acc} \right)$ and we let $T_{stop}$ be the stopping horizon from \cref{prop:multcanstop} for $p$ vehicles; moreover, we define $T_{rec} = K_{rec} \tau \geq T_{stop} + \frac{v_{min}}{u_a} + \frac{d}{v_{min}} + \tau$. We assume that all vehicles of $\mathcal N_t$ for all $t > t_\kappa$ enter safely with a margin $\tau$.

Problem \fh{K_{rec}} is recursively feasible under the hypotheses of \cref{thm:feasible}, \textit{i.e.} if there exists a solution to \fh{K_{rec}} at time $t_\kappa$ for the vehicles of $\NN_{t_\kappa}$, there exists a solution at $t_{\kappa} + \tau$ for the vehicles of $\NN_{t_\kappa + \tau}$.
\end{proposition}
\begin{proof}[Sketch of proof]
The idea between the choice of $T_{rec}$ is to ensure that each vehicle can either stop safely before entering its acceleration region (without generating rear-end collisions), or has already planned to exit its no-stop region safely. Moreover, the safe entering hypothesis ensures that the entry of new vehicles does not invalidate previously safe solutions, which can therefore be extended.
\end{proof}

We obtain the equivalence between \ih{} and \fh{K}:
\begin{theorem}\label{thm:equivalence}
Problems \ih{} and \fh{K} with $K\tau \geq T_{rec}$ are equivalent, \textit{i.e.} an optimal solution to one is also an optimal solution to the other.
\end{theorem}
\begin{proof}
\Cref{prop:ih-to-fh} ensures that any optimal solution to \ih{} is a feasible solution of \fh{K}. \Cref{prop:recursive} shows that a solution to \fh{K} (with $K\tau \geq T_{rec}$) can be recursively extended to a solution of \ih{}; therefore, the optimal solution of \fh{K} is feasible for \ih{}. Using these two results, we deduce the stated theorem.
\end{proof}

An important corollary of \cref{thm:ih-miqp,thm:deadlock,thm:equivalence} is that the control obtained by solving \fh{K} with $K$ large enough is also a solution to the minimally deviating supervision problem, and ensures deadlock avoidance as well. Contrary to \ih{}, \fh{K} is relatively easy to solve with dedicated mixed-integer quadratic programming solvers, as will be demonstrated in the following section.

\section{Simulation Results}\label{sec:simulation}
\subsection{Simulation environment}
The presented Supervisor framework has been validated using extensive computer simulations on various test scenarios. In the absence of standardized test situations and since no open-sourced implementation of comparable methods~\cite{Colombo2014,Campos2015} is available, this section does not aim at a quantitative comparison with existing algorithms. Since our Supervisor is by design guaranteed to output an optimal\footnote{Among the set of piecewise-constant controls with a given time step duration and in the sense of \Cref{def:minimally-deviating}} safe control, the major evaluation criterion is rather its ability to handle a wider variety of traffic scenarios than existing techniques, which is demonstrated in the rest of this section.

Due to implementation reasons, the resolution of the supervision problem is performed off-line and simulations are run in two successive phases. In the first phase, we define the geometry of the roads inside the supervision area and the corresponding possible paths, and compute the collision and acceleration regions information for each pair of paths. Since these sets only depend on the geometry of vehicles and paths, the corresponding parameters are computed off-line.

In the second phase, we run the simulation by coupling the high-fidelity vehicle physics simulator PreScan~\cite{prescan} with an external Python implementation of our supervisor. The actual resolution uses the commercial MIQP solver GUROBI~\cite{gurobi}; the Python program runs a coarse simulation over a set time horizon with a fixed time step duration. Vehicles are generated using random Poisson arrivals, with a predefined arrival rate for each possible path, while respecting the safe entering condition; the initial velocity of each generated vehicle is chosen randomly according to a truncated Gaussian distribution. At each time step, the finite horizon supervision problem \fh{} is solved for the vehicles inside the supervision area, and yields the best safe control for the set of vehicles. The state of these vehicles at the next time step is then computed according to equations~\eqref{eq:cstr-first} and~\eqref{eq:cstr-dyn}.

In parallel, we use PreScan to validate the consistency of this output: from the safe controls computed in the Python supervisor and knowing the reference paths of the vehicles, we compute a target state comprising a desired position, heading and longitudinal velocity for each vehicle. This target state is fed into a low-level controller which outputs a steering and an acceleration or braking control. The vehicle model used in the validation phase takes into account engine response as well as chassis and suspensions dynamics, but does not model road-tire friction. PreScan's collision detection and visualization capacities are then used to validate the absence of collision or dangerous situations. Note that vehicles controllers are designed to ensure a bounded positioning error for any vehicle, relative to their prescribed path and velocity profile. This error is taken into account in the computation of the collision regions, so that the system is robust to control imperfections.

\subsection{Test scenarios}
In the rest of this section, we consider three test scenarios -- chosen to represent a wide variety of driving situations -- consisting of merging on a highway, crossing an intersection or driving inside a roundabout. To showcase the performance of our framework in avoiding accidents and deadlocks, we assume that drivers are ``oblivious'' and focused on tracking a desired speed, regardless of the presence of other vehicles. A video of the presented simulations is available online\footnote{Available at \url{https://youtu.be/JJZKfHMUeCI}}.

\subsubsection{Highway merging}
We first consider a very simple highway merging scenario, where an entry lane merges into a single-lane road; the possible paths for the vehicles are the same as in~\cref{fig:collision-region-c}. The collision region between a vehicle $i$ in the entry lane and a vehicle $j$ on the highway have a single connected component given as $\underline s_{ij}^\perp = \underline s_{ji}^\perp = \SI{89}{\metre}$ and $s_{ij}^\parallel =  s_{ji}^\parallel = \SI{94}{\metre}$, taking control errors into account.

\begin{figure}
	\subfloat[Longitudinal positions; the shaded area is the collision region.]{\includegtikz[width=\columnwidth,height=3.5cm]{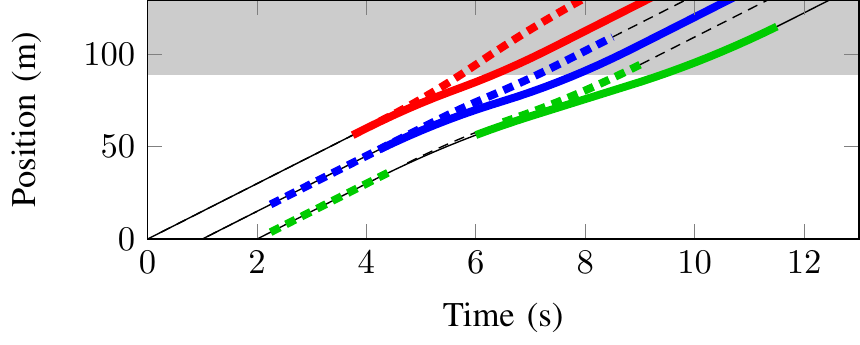}}
	
	\subfloat[Longitudinal velocities]{\includegtikz[width=\columnwidth,height=3.5cm]{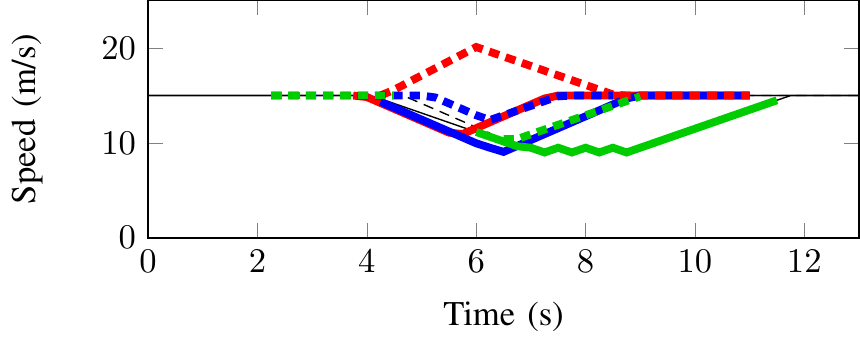}}
	\caption{Vehicles positions and velocities in the merging scenario; solid lines correspond to vehicles on the entry lane, dashed lines to vehicles starting on the highway. The thick colored portions show overriding intervals.\label{fig:merge-sim}}
\end{figure}

To illustrate the action of the supervisor, we consider a set of six vehicles, three of which are on the highway and three on the entry lane. All vehicles are assumed to have ``oblivious'' drivers maintaining a constant speed, thus resulting in potential collisions. This admittedly unrealistic behavior has been chosen to generate a higher probability of collisions in absence of supervision. \Cref{fig:merge-sim} shows the longitudinal trajectories of the supervised vehicles; colored (thick) portions of the lines represent intervals of time during which overriding occurs. The area in gray corresponds to the collision region between entering vehicles and vehicles on the highway; thanks to the action of the supervisor, all collisions are successfully avoided.

\subsubsection{Intersection crossing}
\begin{figure}
	\subfloat[Longitudinal positions; the shaded area is the collision region.]{\includegtikz[width=\columnwidth,height=3.5cm]{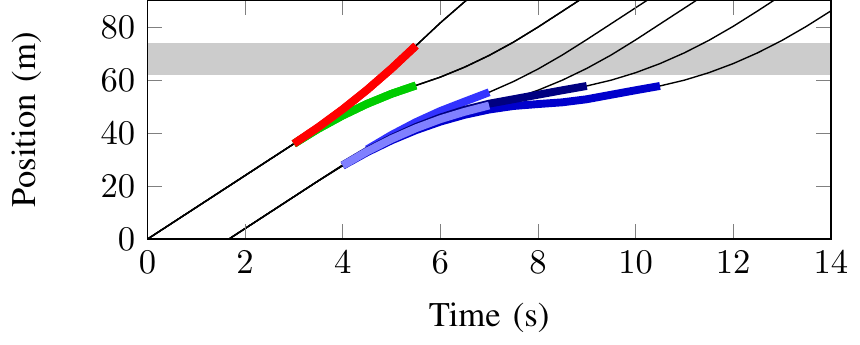}}
	
	\subfloat[Longitudinal velocities]{\includegtikz[width=\columnwidth,height=3.5cm]{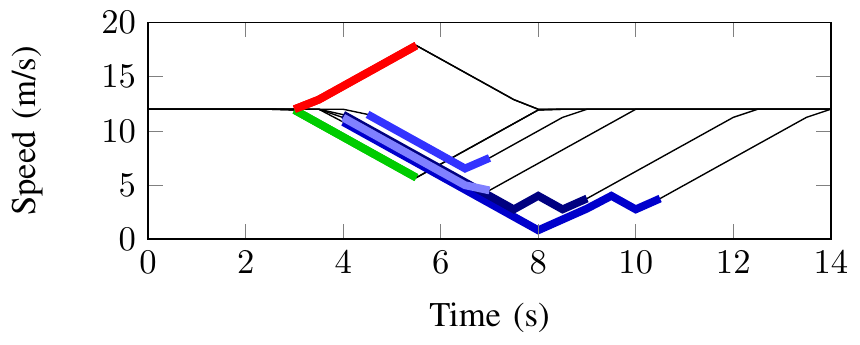}}
	\caption{Vehicles positions and velocities in the intersection crossing scenario; solid lines correspond to vehicles on the entry lane, dashed lines to vehicles starting on the highway. The thick colored portions show overriding intervals.\label{fig:intersec-sim}}
\end{figure}
The second scenario is the crossing of a $+$ shaped intersection by a total of eight vehicles, with two vehicles per branch. In each branch, the front vehicle goes straight, and the rear vehicle turns left; moreover, all vehicles in front start at the same distance from the center of the intersection, and the same is true for the vehicles in the rear. This scenario illustrates the symmetry-breaking capacities of our framework, which handles this perfectly symmetrical scenario well, as shown in~\cref{fig:intersec-sim}. The area in gray corresponds to the collision region between vehicles on different branches. A video of a longer, one hour simulation is available also online\footnote{\url{https://youtu.be/cl32nbceZvw}}.

\subsubsection{Roundabout driving}
\begin{figure}\centering
	\includegtikz[height=5cm]{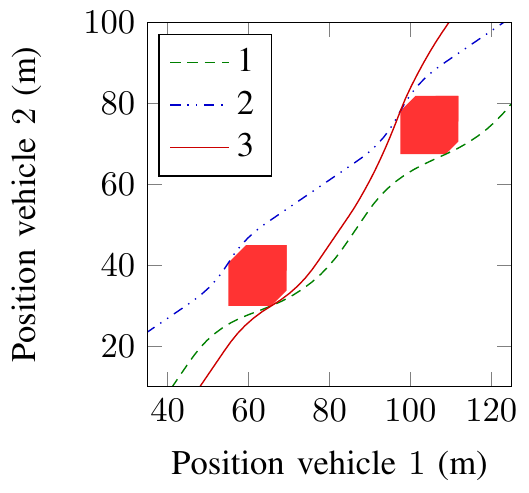}
	\caption{Illustration of three possible classes of trajectories found by the solver, depending on the initial states of the vehicles. Trajectories $1$ and $2$ correspond to one vehicle passing the two collision points before the other. Trajectory $3$ corresponds to the case where the vehicle on the inner lane enters after the other, and overtakes it inside the roundabout. \label{fig:roundabout-sim}}
\end{figure}
Finally, the third scenario consists of vehicles driving inside a two-lanes roundabout. The particularity of this situation is that collision regions can have multiple connected components, for instance for the paths shown in~\cref{fig:collision-region-b}. Since our formulation explicitly distinguishes each of these connected components, the supervisor is able to choose an ordering for each point of conflict, as illustrated in \cref{fig:roundabout-sim}: depending on the initial states and control targets of the vehicles, a different class of solution is chosen. A video of a longer, one hour simulation is also available online\footnote{\url{https://youtu.be/pLoG32wFnkE}}.

\subsubsection{Computation time}
Due to the relatively short time horizon needed to ascertain infinite horizon safety, computation time remains reasonable despite the NP-hardness of the MIQP formulation. \Cref{fig:comput-time} shows the evolution of the computation time in the intersection crossing and roundabout scenarios; the limited available space in the merging scenario does not allow enough vehicles for a similar diagram. These measurements have been obtained on a computer equipped with an Intel Core i7-6700K CPU clocked at \SI{4}{\giga \hertz} with \SI{16}{\giga \byte} of RAM, using the GUROBI solver in version 7.0. It can be seen that computation time remains below the duration of a time step in $90\%$ of cases for up to approximately ten simultaneous vehicles, thus allowing real-time computation at \SI{2}{\hertz}. 

Note that the MIQP problem only loosely depends on the paths geometry, but rather on the average number of conflicts per vehicle which is higher in the case of roundabout driving, thus explaining the longer times reported in~\cref{fig:comput-time-roundabout}. Moreover, the implemented algorithm has been devised for readability over efficiency, and can be optimized by removing redundant variables to further reduce computation time. In practice, this refresh rate means that vehicles could apply a new acceleration every \SI{0.5}{\second}, which is faster than the typical reaction time of one second for a human driver, and should therefore be barely perceived. Note that for practical implementation purposes, the input of the supervisor should be predicted states at the end of the computation period instead of current states; since the acceleration of each vehicle is assumed to be known to the supervisor, these predictions can be easily performed by forward integration.

\begin{figure}\centering
	\subfloat[Intersection scenario]{\includegtikz[width=\linewidth,height=4cm]{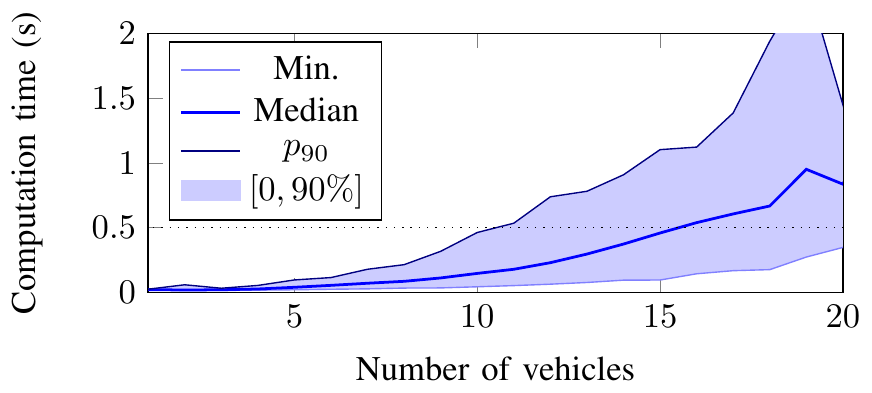}}

	\subfloat[Roundabout scenario]{\includegtikz[width=.9\linewidth,height=4cm]{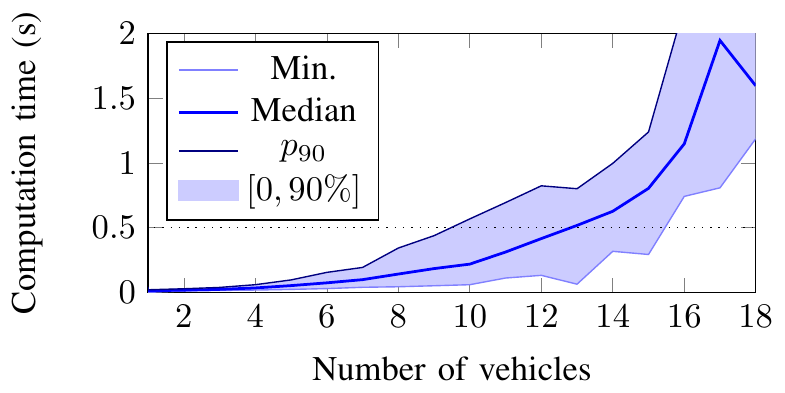}\label{fig:comput-time-roundabout}}
	\caption{Distribution of computation times depending on the number of vehicles, for $\tau = $ \SI{0.5}{\second}. Shaded areas represent the $[0, 90\%]$ percentiles. \label{fig:comput-time}}
\end{figure}

\section{Discussion on implementation}\label{sec:discussion}
In the previous sections, we presented an optimization-based algorithm for the supervision of semi-autonomous vehicles; we now briefly discuss obstacles and possible solutions for actual implementation. First and foremost, not all vehicles will be equipped with the required communication capacities at the same time; therefore, the ability to deal with unequipped vehicles and other traffic participants is key to envision actual applications. Second, this work assumes perfect communication and control, and in general ignores uncertainties arising from real-world constraints.

\subsection{Dealing with unequipped vehicles}
As with all innovations, the penetration rate of our system would gradually increase overtime, but remain below \SI{100}{\percent} for years, yet the formulation proposed in \Cref{sec:miqp-supervised} requires all vehicles to be equipped with supervision capacities. Although a detailed study on the integration of unequipped vehicles in our framework is out of the scope of this paper, we present a possible technique to handle these vehicles provided that they can avoid longitudinal collisions with the leading vehicle, and have a bounded reaction time.

First, note that it is always possible to consider unequipped vehicles conservatively as proposed in~\cite{Schildbach2016}: at a given step $k$, we compute the minimum and maximum curvilinear position that can be reached at time $t_k$ by the unequipped vehicle $i_u$, denoted by $s_{i_u,min}^k$ and $s_{i_u,max}^k$ respectively. Using the same notations as in \Cref{sec:miqp-supervised}, we then define:
\begin{align}
\varepsilon_{{i_u}j,p}^\parallel(k) = &\ \chi_{[s_{{i_u}j,p}^\parallel,+\infty[}(s_{i_u,max}^k), \\
\varepsilon_{{i_u}j,p}^\perp(k) = &\ \chi_{[\overline s_{{i_u}j,p}^\perp,+\infty[}(s_{i_u,min}^k).
\end{align}
Therefore, the unequipped vehicle is considered as occupying the conflict region at step $k$ when there exists a control (maximum acceleration) for which it could be inside this region at step $k$. Similarly, the vehicle is only considered as liberating the conflict region when, even by applying a maximum braking, it would exit it. The collision avoidance constraints \eqref{eq:safety-2} and \eqref{eq:safety-3} are also modified to use $s_{i_u,min}^{k+1}$ and $v_{i_u,min}^{k+1}$, where $v_{i_u,min}^{k+1}$ is the minimum speed reachable by $i_u$ at $k+1$. Other traffic participants such as cyclists (and, to a lesser extent, pedestrians) could also be taken into account in this fashion. Recently proposed ``non-conservatively defensive strategies''~\cite{Zhan2016} could also be applied.

A limitation of this simple approach is that it can lead equipped vehicles to often yield right-of-way to unequipped vehicles, which may be problematic and can slow the acceptance of the system. A possible method (introduced in~\cite{Qian2014}) to reduce this problem while improving the global level of safety is to use the existing equipped vehicles to force the unequipped ones to stop when required. Suppose that an unequipped vehicle (denoted by $i_u$) follows an equipped one ($i_e$), both crossing the path of another equipped vehicle $j_e$. By setting $\pi_{{i_e}{j_e}} = 0$ (thus requiring $j_e$ to pass before $i_e$), we effectively force the unequipped vehicle $i_u$ to also pass after $j_e$; the reaction time of the unequipped vehicle can be taken into account by adjusting the lower bound on the longitudinal acceleration of vehicle $i_e$.

Note that this approach still guarantees that no collision can happen between an unequipped and an equipped vehicle; moreover, as the penetration rate of equipped vehicles increases, additional rules may be enforced to reduce the number of occurrences in which conflicting unequipped vehicles are simultaneously allowed in the conflict region, thus increasing safety even for the unequipped vehicles. Future work will study the impact of penetration rate on safety and efficiency for both equipped and unequipped vehicles.

\subsection{Practical implementation}
We propose a centralized implementation, where a roadside computer (\textit{supervisor}) with communication capacities is added to the infrastructure, and is tasked with repeatedly solving \fh{K}. Note that resolution could also be performed using cloud computing, possibly providing much faster computations without necessitating fully dedicated hardware. The supervisor is also assumed to be equipped with a set of sensors (\textit{e.g.}, cameras), so that the arrival of new vehicles in the supervision area can be monitored (in order to account for unequipped vehicles and other traffic participants). Equipped vehicles are supposed to regularly communicate their current state, including position, velocity and driver's control input, and receive instructions (the safe acceleration sequence $(u_i^k)$ solution of \fh{K}) from the roadside supervisor. The vehicle's on-board computer then uses these instructions to override the driver's control inputs when needed. We argue that the main sources of uncertainty, \textit{i.e.} communication, sensing and control errors, can be taken into account by using safety margins when computing collision regions.

Communications are assumed to have similar performance to current 802.11p specifications; we use the figures provided in~\cite{Jiang2006,Demmel2012} as reference, with latency below \SI{20}{\milli\second}, and packet loss probability of less than \SI{30}{\percent} under \SI{300}{\meter}. To account for network congestion, we use more conservative values than those reported experimentally in~\cite{Demmel2012}. Moreover, using the additional roadside sensors, we estimate that uncertainty in each vehicle's localization could be reduced to below \SI{1}{\meter} longitudinally.

First, the \SI{20}{\milli\second} latency corresponds to less than \SI{1}{\meter} at highway speed. Second, since they do not require exchanging a lot of data, such messages can be sent much more frequently than the refresh rate of the supervisor. Considering messages can be sent at \SI{20}{\hertz}, the probability of a message not being received in \SI{0.25}{\second} is roughly \SI{0.2}{\percent}, and \SI{6e-6}{} after \SI{0.5}{\second}. Since they receive a whole sequence of safe accelerations, individual vehicles can keep executing this sequence until a new one is successfully received. A worst-case scenario would be having one vehicle using acceleration $u_a$ (maximum acceleration) where it should have used $u_b$ (maximum braking): after a duration $t$, the corresponding positioning error is $\frac 1 2 t^2 (u_a + |u_b|)$, which is roughly \SI{30}{\centi\meter} after \SI{0.25}{\second} and \SI{1.3}{\meter} after \SI{0.5}{\second} for typical values of $u_a$ and $|u_b|$ of \SI{5}{\meter \per \second\squared}. More robust contingency protocols could likely be developed, and will be the subject of future work, but these values can be used as safety margins without compromising performance.

Similarly, positioning and control uncertainty can be accounted for as margins in the collision regions, provided they can be bounded. In this work, we assume that vehicle self-positioning can be improved using the roadside sensors from the supervisor (which can be precisely calibrated), which could provide relatively tight bounds on error.

\section{Conclusion}\label{sec:conclusion}
In this article, we designed a framework allowing safe semi-autonomous driving of multiple cooperative vehicles in various traffic situations. We first introduce a set of linear constraints ensuring infinite horizon safety for a group of human-driven vehicles, traveling inside predefined corridors with the help of existing lane-keeping technologies. Based on this set of constraints, a discrete-time \textit{Supervisor} is allowed to override the drivers' longitudinal control inputs if they would lead the vehicles into an inevitable collision state. In this case, the control used for overriding is chosen as close as possible to the one originally requested by the drivers. These two properties ensure that intervention only occurs when strictly necessary to maintain safety, thus facilitating the acceptation of the system by human drivers.

Theoretical considerations prove this supervisor guarantees both safety and deadlock avoidance, and can be applied without distinction to multiple situations such as traffic intersection, highway entry lanes or roundabouts. Using the realistic vehicle physics simulator PreScan, we demonstrated that our algorithm can handle complex situations over an arbitrary duration, with continuous arrivals of vehicles. Moreover, the proposed formulation can be solved in real-time on a standard desktop computer for up to ten vehicles, which makes it suitable for practical applications.

Additionally, this work opens up many perspectives for future research. First and foremost, the current framework does not deal with non-equipped vehicles or other traffic participants such as cyclists or pedestrians, nor does it take sensor and communication uncertainties into account. Before considering an actual implementation, the system should be more robust to these various sources of noise. Moreover, our formulation has been designed in a mostly centralized fashion; various approaches need to be explored to design a more realistic decentralized system, that could be implemented in actual cars.


%

\ifapp
\appendices
\crefalias{section}{appsec}
\crefalias{subsection}{appsec}

\section{Demonstrations}
\subsection{Proofs for \Cref{sec:miqp-supervised}}\label{app:proof-deadlock}
Before proving \cref{thm:deadlock}, we introduce the following lemma stemming from graph theory:
\begin{lemma}\label{lem:graph}
Let $\mathcal G = (V, E)$ a directed graph with vertices set $V$ and edges set $E$. All cycles in $\mathcal G$ can be removed by reversing a set of edges, each of them contributing to at least one cycle.
\end{lemma}
\begin{proof}
The proof is based on the existence of minimum feedback arc sets~\cite{Younger1963}, \textit{i.e.} a minimum set $E_{feedback} \subset E$ such that $\mathcal G' = (V, E \setminus E_{feedback})$ is acyclic. By minimality of $E_{feedback}$, any $e \in E_{feedback}$ belongs to at least one cycle of $\mathcal G$. Moreover, it can be seen that reversing the edges of $E_{feedback}$ also leads to an acyclic graph, thus proving the lemma.
\end{proof}

\begin{proof}[Proof of \Cref{thm:deadlock}]
	Note that the only constraints requiring a vehicle to stop are~\eqref{eq:safety-1} to~\eqref{eq:safety-3}, forcing a vehicle $j$ to wait for a vehicle $i$ with $\pi_{ij}^p = 1$. From the hypotheses and \cref{thm:feasible}, there exists a solution $\mathbf X$ to~\ih{} at time $t_{\kappa_0}$. We define a directed \textit{priority graph} $\mathcal G_{\mathbf X} = (V,E)$ with $V = \NN_{t_{\kappa_0}}$ and where an edge $i \rightarrow j$ belongs to $E$ if there exists $p$ such that $\pi_{ij}^p = 1$. Using this representation, a cycle in $\mathcal G_{\mathbf X}$ corresponds to a chain of $q$ conflicting vehicles $i_1, i_2, \dots, i_q, i_{q+1} = i_1$ for which there exists a connected component $\mathcal C_{i_n i_{n+1}}^{p_n}$ such that $\pi_{i_ni_{n+1}}^{p_n} = 1$ for all $n = 1\dots q$. 

	If $\mathcal G_{\mathbf X}$ is acyclic, it defines a (partial) topological order, and it is always possible to admit the vehicles one by one in that order. Therefore, there exists a feasible solution where all the vehicles of $\NN_{t_{\kappa_0}}$ exit the supervision area in finite time.

	We now assume that there exists at least one cycle in $\mathcal G_{\mathbf X}$. If all the vehicles involved in the cycle can exit in finite time, the result of the theorem is proven. Otherwise, we note $\NN_{dead} \subset \NN_{t_{\kappa_0}}$ a set of vehicles corresponding to a cycle in $\mathcal G_{\mathbf X}$: all of these vehicles are stopped at infinity, and are prevented to move further by a constraint of form~\eqref{eq:safety-1}, for a certain $j \in \NN_{dead}$. Moreover, the no-stop condition~\eqref{eq:cstr-deadlock} ensures that, for all $i \in \NN_{dead}$ and all $k \geq \kappa$, $s_i^k \leq \underline s_i^\perp$. 

	From \cref{lem:graph}, we know that it is possible to change the values of the variables $\pi_{ij}^p$ for $i, j \in \NN_{dead}$ to render $\mathcal G_{\mathbf X}$ acyclic. Using the fact that $s_i^k \leq \underline s_i^\perp$ for all of these vehicles, we know that modifying these priorities does not violate constraints~\eqref{eq:safety-1} to~\eqref{eq:safety-3}. Therefore, we can build a solution $\mathbf X'$ for which the corresponding priority graph is acyclic, which proves the theorem.
\end{proof}

\subsection{Proofs for \Cref{sec:finite-horizon}}\label{app:proofs-fh}
\begin{proof}[Proof of \Cref{lem:canstop}]The proof is trivial if we consider continuous-time dynamics, as a single vehicle can always apply a control lower or equal to $u_b$ starting from time $t_k + \tau$, which ensures it is stopped for $t \geq t_k + \tau + \frac{v_{max}}{u_b}$. A slight additional complexity happens at the final braking time step when considering piecewise-constant controls, applying $u_b$ for a duration $\tau$ might result in a negative velocity, which is not allowed in our framework. We now proceed to the formal proof, as below.
	
Let $(u_i^k)$ be the control corresponding to trajectory $s_i$, and let us define a control $(w_i^k)$ as: $w_i^{\kappa} = u_i^{\kappa}$, $w_i^k = \min(u_{b}, u_i^k)$ for $\kappa < k < \kappa + K$, and $w_i^{\kappa+K} = u_{b}$. We construct $(\tilde u_i^k)$ iteratively as $\tilde u_i^{\kappa} = u_i^k$ and, for $k \geq \kappa+1$, $\tilde u_i^k = \left\{ \begin{array}{l l}w_i^k & \mathrm{if\ } \tilde v_i^k + w_i^k\tau \geq 0\\ -\frac{v_i^k}{\tau} & \mathrm{otherwise}\end{array}\right.$, where $\tilde v_i^k$ is the speed of vehicle $i$ at time $t_k$ under control $(\tilde u_i^k)$.

As $\tilde v_i^{\kappa+1} = v_i^{\kappa+1} \leq v_{max} \leq (K-1)\tau |u_{b}|$ from the hypothesis, there exists a minimal value of $k_0 \geq \kappa$ such that $\tilde v_i^{k_0} \leq |u_{b}| \tau$; moreover, the condition on $K$ ensures that $\tilde v_i^{\kappa + 1} - (K-2)|u_b|\Delta_t \leq |u_b|\Delta_t$, and so $k_0 \leq (\kappa + 1) + (K - 2) = \kappa + K - 1$. 

From the definition of $(\tilde u_i^k)$, we know that for all $k_0 + 1 \leq k \leq \kappa + K$, $\tilde v_i^{k} = 0$. Since $\tilde u_i^k \leq u_i^k$ for $\kappa \leq k \leq k_0-1$, we also know that $\tilde s_i^{k_0} \leq s_i^{k_0}$ and $\tilde v_i^{k_0} \leq v_i^{k_0}$. Finally, $\tilde u_i^{k_0}$ is the minimal admissible control starting from $\tilde x_i^{k_0}$; therefore, $\tilde s_i^{k_0+1} \leq s_i^{k_0+1} = s_i^{\kappa+K}$ which proves the above lemma.
\end{proof}

\begin{proof}[Proof of \Cref{prop:multcanstop}] We first consider the continuous-time case to give an intuition of the proof. We build upon the fact that the rearmost vehicle in a line can always brake with acceleration $u_b$ until it fully stops. However, some the initial conditions may require vehicles in front to accelerate in order to avoid collisions, for instance if the rearmost vehicle is too fast. However, even in this case, we know that the second rearmost vehicle can brake with $u_b$ as soon as it has matched the speed of the rearmost vehicle, and by induction this is true for all the vehicles in the line. In the continuous-time case, all of these vehicles can therefore stop in a time bounded by $\frac{v_{max}}{|u_b|}$.
	
When considering piecewise-continuous controls, an additional complexity arises from the fact that vehicles may match speed between two time steps, resulting in an ``overshoot'' in velocity. We use the hypotheses on the acceleration to bound this overshoot, as illustrated in \cref{fig:overshoot}: we consider a fast vehicle (noted $1$, blue curve) following a slower vehicle (noted $2$, red curve). To avoid collisions, vehicle $2$ is required to accelerate; vehicles match speed at time $t_{2\rightarrow 1}$; however, due to the time discretization, the overshoot phenomenon can occur. Using the bounds on the acceleration, noting $k_{2\rightarrow 1}$ the time step immediately following $t_{2\rightarrow 1}$, we know that $v_2^k \leq v_1^k + \tau(u_{max} + |u_b|)$. Moreover, after step $k$, vehicle $2$ can brake with a control $u_b$ up to time $t_1$, corresponding to the first integer time step $k_1$ when $v_1^{k_1} \leq \tau |u_b|$. Since both vehicles $1$ and $2$ have the same acceleration of $[t_{2\rightarrow 1}, t_1]$, we know that $v_2^{k_1} \leq \tau |u_b| + \tau (|u_b| + u_{max})$. Therefore, noting $k_2 = k_1 + 1 + \ceil{\frac{u_{max}}{|u_b|}}$, we know that $v_2^{k_2} \leq \tau |u_b|$. The same reasoning can then be repeated for the vehicles preceding vehicle $2$. We will now formalize this recursion, as below.

\begin{figure}
	\includegtikz[width=\linewidth,height=4cm]{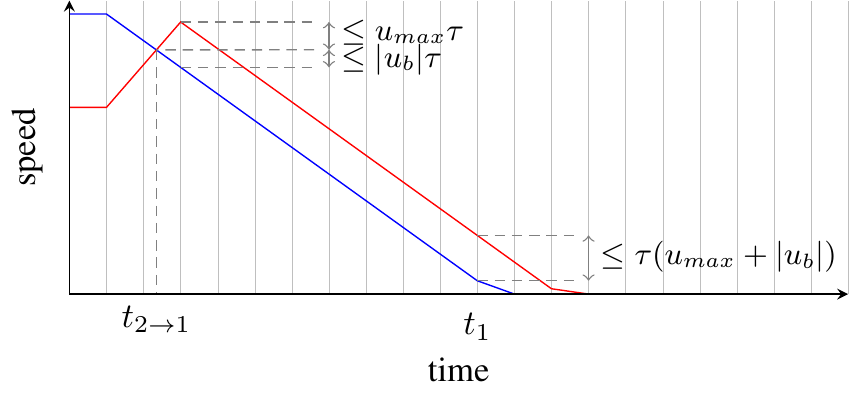}
	\caption{Illustration of the ``overshoot'' phenomenon; the vertical grid correspond to integer multiples of the time step.\label{fig:overshoot}}
\end{figure}
	
We will prove by induction that, for $i \in \{1,\dots,p\}$, there exists a dynamically feasible control $(\hat u_i^k)$ with $\hat u_i^\kappa = u_i^\kappa$ and $\hat u_i^k \leq u_i^k$ for $k \geq \kappa$ such that the corresponding vehicle speed $(\hat v_i^k)$ verifies $\hat v_i^{\kappa + K_i} \leq |u_b|\tau$ with $K_i \tau = \ceil{\frac{v_{max}}{|u_b|}} + (i-1) \left(1+\ceil{\frac{u_{max}}{|u_b|}}\right) \tau$. First, for the rearmost vehicle $i=1$, the proof of \cref{lem:canstop} provides the result with $\hat u_i^k = \tilde u_i^k$.

We now let $i \geq 2$ and assume that every vehicle $j \in \{1, \dots, i-1\}$ follows its corresponding control $(\hat u_j^k)$. We note $\hat s_j^k$ and $\hat v_j^k$ the position and speed of vehicle $j$ at step $k$ under this control. Since $\hat u_j^k \leq u_j^k$ for these vehicles, we deduce from the monotony of the system that the original control solution for vehicle $i$ $(u_i^k)_{\kappa \leq k < \kappa + K}$ prevents rear-end collisions if vehicle $i-1$ applies $\hat (u_{i-1}^k)$. Therefore, any dynamically feasible extension of $(u_i^k)$ is safe over $[\kappa \tau, (\kappa + K + 1) \tau[$. As a result, the set $U_i^{safe}(u_i^{\kappa},[\kappa,\kappa+K])$ of all admissible controls $(\hat u_i^k)_{\kappa \leq k \leq \kappa+K}$ for vehicle $i$ such that $\hat u_i^{\kappa} = u_i^{\kappa}$ and $\underline u_i^k \leq \hat u_i^k \leq u_i^k$ for $\kappa \leq k < \kappa+K$ is not empty. We note $(\hat u_i^k)$ a minimum element of this set (and so $\hat u_i^k \leq u_i^k$); we will prove that $\hat v_i^{\kappa+ K_i} \leq |u_b|\tau$.

If for all $k \geq \kappa$, $\hat v_i^k \leq \tilde v_{i-1}^k$, we conclude that vehicle $i$ stops before vehicle $i-1$ which proves the result from the induction hypothesis. Otherwise, we let $k_0^i \geq \kappa$ be the minimum $k$ such that $\hat v_i^k \geq \tilde v_{i-1}^k$, and we know that $\hat v_i^{k_0^i} \leq \hat v_{i-1}^{k_0^i-1} + u_{max} \tau$. For all $k \geq k_0^i$, we know from the monotony of the system that the control $\min(\hat u_{i-1}^k, u_i^k, u_b)$ prevents rear-end collisions; we deduce that, for $k \geq k_0^i$, $\hat v_i^k \leq \hat v_{i-1}^{k-1} + u_{max} \tau$. Therefore, $\hat v_i^{\kappa + K_{i-1} + 1} \leq \hat v_{i-1}^{K_{i-1}} + u_{max} \tau$ and we deduce from the induction hypothesis that $\hat v_i^{\kappa + K_{i-1}+1} \leq u_{max} \tau + |u_b| \tau$. Therefore, we obtain the recursion relation $K_i = K_{i-1} + 1 + \ceil{\frac{u_{max}}{|u_b|}}$ which yields the announced result.

Finally, we conclude that vehicle $i$ can fully stop (without rear-end collisions) at step $\kappa + K_i + 1$; therefore the set of $p$ vehicles can safely stop before the beginning of step $\kappa + K$ if $K\geq K_p = \frac{v_{max}}{|u_b|\tau} + (p-1) \ceil{\frac{u_{max}}{|u_b|}} + 1$. Since the recursion ensures that for all $i$ and $k$, $\hat u_i^k \leq u_i^k$, we deduce that $s_i^{\kappa + K} \leq \hat s_i^{\kappa + K}$ which proves the proposition.

Note that the time needed for vehicles to match speeds can also be bounded by $\frac{v_{max}}{u_a}$ regardless of the value of $p$. Therefore, all vehicles can also fully stop within the time horizon if $T = K\tau \geq \frac{v_{max}}{|u_b|} + \frac{v_{max}}{u_a} + 2\tau$. Depending on the value of $p$, this bound may be better than the previously demonstrated one.
\end{proof}

\begin{proof}[Proof of \Cref{prop:recursive}]
	We consider a time $t_\kappa = \kappa \tau$, and we let $T_{rec} = K_{rec}\tau \geq T_{stop} + \frac{v_{min}}{u_a} + \frac{d}{v_{min}} + \tau$. Consider a solution $\mathbf X$ of \fh{K_{rec}} for the vehicles of $\NN_{t_\kappa}$, defined for steps $\kappa \leq k \leq \kappa + K$. We will first show that this solution can be extended to a solution of \fh{K+1} for the vehicles in $\NN_{t_\kappa}$. Note that the only constraints which can be unfeasible are the safety constraints~\eqref{eq:safety-1}-\eqref{eq:safety-3} and the minimum velocity constraints~\eqref{eq:cstr-accel} and~\eqref{eq:cstr-deadlock}. Consider a vehicle $i \in \NN_{t_\kappa}$: using the control corresponding to this solution, two cases can arise:
	\begin{itemize}
	\item $s_i(T_{stop}) \leq s_i^{acc}$, in which case \cref{prop:multcanstop} ensures that $i$ and all the vehicles behind it can fully stop before reaching $s_i^{acc}$, and can remain stopped up to step $K_{rec} + 1$. Since we also require that $s_j^{acc} \geq s_i^{acc}$ if $j$ follows $i$, this ensures that keeping $i$ and its followers stopped satisfies all the above constraints;
	\item otherwise, $s_i(T_{rec}) \geq \overline s_i^\perp$, in which case the crossing and minimum velocity constraints~\eqref{eq:safety-1},~\eqref{eq:cstr-accel} and~\eqref{eq:cstr-deadlock} are satisfied for all conflicting vehicle $j$ up to step $K_{rec}$. The requirement $T_{rec} \geq T_{stop}$ and \cref{prop:multcanstop} ensure that the safe following constraints~\eqref{eq:safety-2} and~\eqref{eq:safety-3} involving vehicle $i$ remain satisfiable for the vehicles of $\NN_{t_\kappa}$ up to step $K_{rec} + 1$.
	\end{itemize}
	Indeed, if $\underline s_i^\perp \geq s_i(T_{stop}) > s_i^{acc}$, condition~\eqref{eq:cstr-accel} ensures that vehicle $i$ accelerates at least with acceleration $u_a$ until reaching speed $v_{min}$, which takes at most a time $\frac{v_{min}}{u_a}$. The vehicle is then required to maintain speed $v_{min}$ until reaching $\overline s_i^\perp$, which takes at most a time $\frac{d}{v_{min}}$. Therefore, vehicle $i$ necessarily reaches $\overline s_i^\perp$ by time $T_{rec}$; the additional $\tau$ accounts for vehicles reaching doing so between two time steps.

	The above considerations ensure that the solution of \fh{K} at time $t_\kappa$ can be prolongated to a solution of \fh{K+1} for the vehicles of $\NN_{t_\kappa}$. Finally, noting that the safe entry hypothesis ensures that this solution remains safe even when taking the vehicles of $\NN_{t_\kappa + \tau} \setminus \NN_{t_\kappa}$ into consideration. By definition, there also exists a safe control (and therefore a solution to \fh{K}) for these vehicles at time $t_\kappa + \tau$. As a result, there exists a solution to \fh{K} at time $t_\kappa + \tau$ for the vehicles of $\NN_{t_\kappa + \tau}$ which proves the stated result.
\end{proof}
\fi

\ifCLASSOPTIONcaptionsoff
  \newpage
\fi



\bibliographystyle{IEEEtran}
\bibliography{MILP_Supervisor}

%

\begin{IEEEbiography}[{\includegraphics[width=1in,height=1.25in,clip,keepaspectratio]{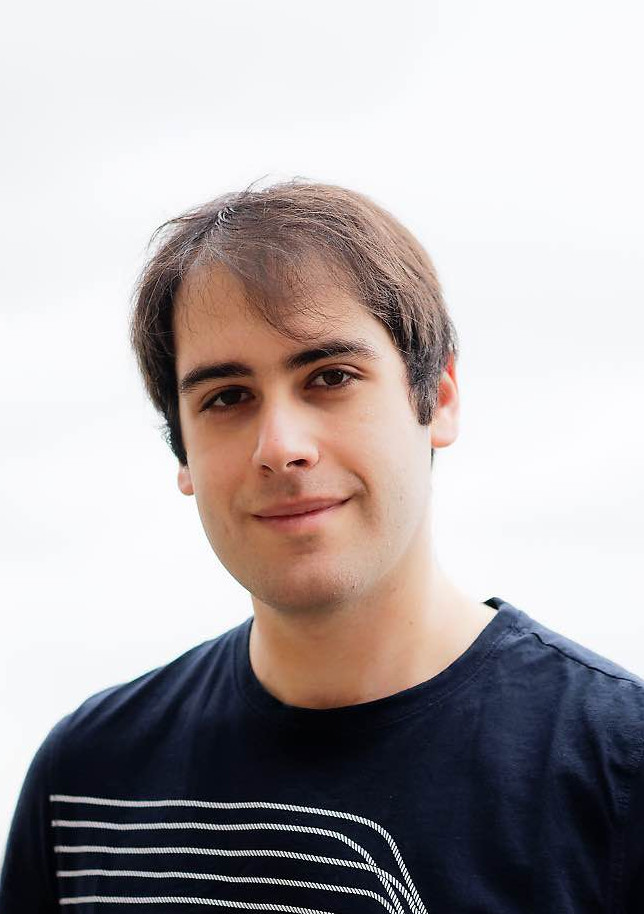}}]{Florent Altch\'e}
received the M.S. degree in engineering from French \'Ecole Polytechnique in 2014 and a Specialized Master degree of public policies from \'Ecole des Ponts ParisTech in 2015. He is currently working towards a Ph.D. degree at the Centre for Robotics at Mines ParisTech. His research interests include coordination and cooperation of autonomous vehicles or robots, as well as planning in an uncertain environment.
\end{IEEEbiography}

\begin{IEEEbiography}[{\includegraphics[width=1in,height=1.25in,clip,keepaspectratio]{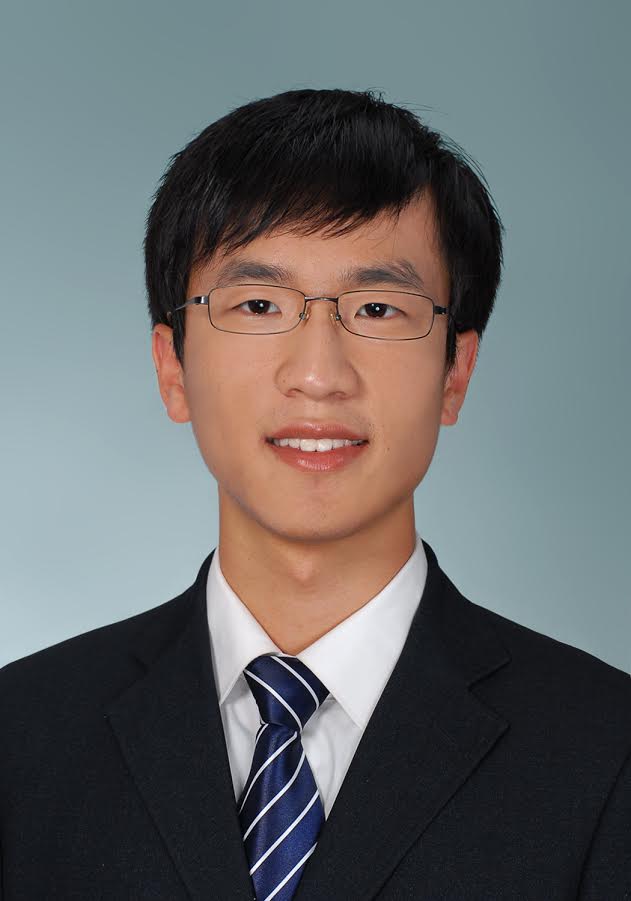}}]{Xiangjun Qian}
received the B.S. degree in computer science from Shanghai Jiao Tong University, Shanghai, China in 2010, the Dip-Ing degree from MINES ParisTech, Paris, France, in 2012, and received his Ph.D. degree in robotics and automation at MINES ParisTech in 2016. His main research interest lies in the control and coordination of autonomous vehicles. He is also interested in the application of machine-learning techniques on the analysis of large-scale transportation networks.
\end{IEEEbiography}


\begin{IEEEbiography}[{\includegraphics[width=1in,height=1.25in,clip,keepaspectratio]{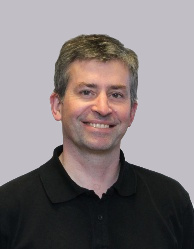}}]{Arnaud de La Fortelle}
(M’10) received the M.S. degree in engineering from \'Ecole Polytechnique and \'Ecole des Ponts et Chaussées, Paris, France, in 1994 and 1997 respectively and the Ph.D. degree from \'Ecole des Ponts et Chaussées in 2000, with a specialization in applied mathematics. He is professor at MINES ParisTech, Paris, France.

From 2003 to 2005, he investigated communications for cooperative systems and the architecture required in distributed systems at INRIA, participating in the CyberCars project. Since 2006, he has been the director of the Joint Research Unit LaRA (La Route Automatis\'ee) of INRIA and MINES ParisTech and, since 2008, has also served as the Director of the Center of Robotics in MINES ParisTech. 
His research interests include cooperative systems (communication, data distribution, control, and mathematical certification) and their applications (Autonomous vehicles, collective taxis...). He coordinates the international research chair \textit{Drive for All} (with partners UC Berkeley, Shanghai JiaoTong University and EPFL).

Dr. de La Fortelle has been elected to the Board of Governors of IEEE Intelligent Transportation System Society in 2009 and is member of the Board of the French Automotive Engineers Society. He is President of the French ANR evaluation committee for sustainable transport and mobility since 2015 and served as expert in European H2020 program.
\end{IEEEbiography}




\end{document}